\pdfoutput=1
\documentclass[11pt]{article}

\let\originalleft\left
\let\originalright\right
\renewcommand{\left}{\mathopen{}\mathclose\bgroup\originalleft}
\renewcommand{\right}{\aftergroup\egroup\originalright}

\usepackage[colorlinks, allcolors=blue, linktocpage, breaklinks]{hyperref}
\usepackage{amsmath, amssymb, amsthm}
\usepackage[capitalize]{cleveref}
\usepackage{enumerate,enumitem}
\usepackage{autobreak, bbm, caption, geometry, graphicx, mathtools, sidecap, tikz, tikzsymbols, titling, verbatim, xspace}

\usepackage[backend=biber, isbn=false, backref, maxbibnames=99]{biblatex}
\DefineBibliographyStrings{english}{
	backrefpage = {p\adddot},
	backrefpages = {pp\adddot}
}

\allowdisplaybreaks

\theoremstyle{plain}
\newtheorem{thm}{Theorem}[section]
\newtheorem{lem}[thm]{Lemma}
\newtheorem{cor}[thm]{Corollary}
\newtheorem{prp}[thm]{Proposition}

\newtheorem{clm}[thm]{Claim}

\newtheorem*{rst*}{Reminder}

\crefname{lem}{Lemma}{Lemmas}
\crefname{thm}{Theorem}{Theorems}
\crefname{cor}{Corollary}{Corollaries}

\theoremstyle{definition}
\newtheorem{dfn}[thm]{Definition}

\newtheorem{prb}[thm]{Problem}
\crefname{prb}{Problem}{Problems}

\newtheorem{rmk}[thm]{Remark}

\theoremstyle{remark}
\newtheorem*{rmk*}{Remark}

\newcommand*{\C}{\ensuremath{\mathbb{C}}}
\newcommand*{\N}{\ensuremath{\mathbb{N}}}
\newcommand*{\R}{\ensuremath{\mathbb{R}}}

\newcommand*{\bra}[1]{\ensuremath{\langle #1|}}
\newcommand*{\ket}[1]{\ensuremath{|#1\rangle}}

\newcommand*{\ip}[2]{\langle #1 | #2 \rangle}
\newcommand*{\ketbra}[2]{|#1\rangle\!\langle #2|}
\newcommand*{\kb}[1]{\ketbra{#1}{#1}}

\newcommand*{\acz}{$\mathrm{AC}^0$\xspace}
\newcommand*{\qacz}{$\mathrm{QAC}^0$\xspace}
\newcommand*{\qacf}{$\mathrm{QAC}_{\mathrm f}$\xspace}
\newcommand*{\qaczf}{$\mathrm{QAC}^0_{\mathrm f}$\xspace}

\newcommand*{\iv}{\phi} 
\newcommand*{\aov}{\varphi} 
\newcommand*{\dov}{\psi} 
\newcommand*{\cv}{\theta} 

\usetikzlibrary{shapes,arrows,quantikz}
\newcommand*{\orctrl}{\ctrl}

\newcommand*{\adj}[1]{#1^\dagger}
\DeclareMathOperator*{\argmin}{argmin}
\newcommand*{\bern}{\mathrm{Bernoulli}}
\newcommand*{\bits}{\{0,1\}}
\newcommand*{\eps}{\varepsilon}
\newcommand*{\Ex}{\mathbb{E}}
\newcommand*{\gauss}{\mathcal G}
\newcommand*{\Hi}{\mathcal H}
\newcommand*{\Ind}[1]{\mathbbm1_{#1}}
\newcommand*{\norm}[1]{\|#1\|}
\newcommand*{\Norm}[1]{\left\|#1\right\|}
\newcommand*{\poly}{\mathrm{poly}}
\newcommand*{\rt}[1]{R_{#1}}
\newcommand*{\rtt}{$\rt\otimes$\xspace}
\newcommand*{\sscat}{\text{\scriptsize\Cat}}
\newcommand*{\zs}{0\dotsc0}
\DeclarePairedDelimiter{\ceil}{\lceil}{\rceil}

\title{Bounds on the \qacz Complexity of Approximating Parity}
\author{Gregory Rosenthal\thanks{Email: \href{mailto:rosenthal@cs.toronto.edu}{\color{black}\texttt{rosenthal@cs.toronto.edu}}. Supported by NSERC (PGS D).} \\ University of Toronto}
\date{\vspace{-5ex}}

\begin{document}
\begin{NoHyper}
	\maketitle
\end{NoHyper}
\begin{abstract}
	QAC circuits are quantum circuits with one-qubit gates and Toffoli gates of arbitrary arity. \qacz circuits are QAC circuits of constant depth, and are quantum analogues of \acz circuits. We prove the following:
	\begin{itemize}[leftmargin=*]
		\item For all $d \ge 7$ and $\eps>0$ there is a depth-$d$ QAC circuit of size $\exp(\poly(n^{1/d}) \log(n/\eps))$ that approximates the $n$-qubit parity function to within error $\eps$ on worst-case quantum inputs. Previously it was unknown whether QAC circuits of sublogarithmic depth could approximate parity regardless of size.
		\item We introduce a class of ``mostly classical" QAC circuits, including a major component of our circuit from the above upper bound, and prove a tight lower bound on the size of low-depth, mostly classical QAC circuits that approximate this component.		
		\item Arbitrary depth-$d$ QAC circuits require at least $\Omega(n/d)$ multi-qubit gates to achieve a $1/2 + \exp(-o(n/d))$ approximation of parity. When $d = \Theta(\log n)$ this nearly matches an easy $O(n)$ size upper bound for computing parity exactly.
		\item QAC circuits with at most two layers of multi-qubit gates cannot achieve a $1/2 + \exp(-o(n))$ approximation of parity, even non-cleanly. Previously it was known only that such circuits could not cleanly compute parity exactly for sufficiently large $n$.
	\end{itemize}
	The proofs use a new normal form for quantum circuits which may be of independent interest, and are based on reductions to the problem of constructing certain generalizations of the cat state which we name ``nekomata" after an analogous cat y\=okai.
\end{abstract}
\newpage
\tableofcontents
\section{Introduction}\label{intro}
\subsection{Background}\label{back}

A central problem in computational complexity theory is to prove lower bounds on the nonuniform circuit size required to compute explicit boolean functions. Since this appears to be out of reach given current techniques, research in circuit complexity has instead focused on proving lower bounds in restricted circuit classes. There are now many known lower bounds in \emph{classical} circuit complexity, as well as in quantum \emph{query} complexity, but comparatively few lower bounds are known in quantum circuit complexity, which is the subject of the current paper.

The study of quantum circuit complexity was initiated in large part by Green, Homer, Moore and Pollett~\cite{Gre+02}, who defined quantum analogues of a number of classical circuit classes. One of the seemingly most restrictive quantum circuit classes that they defined is the class of \qacz circuits, consisting of constant-depth QAC circuits, where QAC circuits are quantum circuits with arbitrary one-qubit gates and generalized Toffoli gates of arbitrary arity. (More precisely, $(n+1)$-ary generalized Toffoli gates are defined by $\ket{x,b} \mapsto \ket{x, b \oplus \bigwedge_{j=1}^n x_j}$ for $x = (x_1, \dotsc, x_n) \in \bits^n, b \in \bits$.) This is analogous to the classical circuit class of \acz circuits, consisting of constant-depth AC circuits, where AC circuits are boolean circuits with NOT gates and unbounded-fanin AND and OR gates. Low-depth circuits are a model of fast parallel computation, and this is especially important for quantum circuits, because quantum computations need to be fast relative to the decoherence time of the qubits in order to avoid error.

One difference between AC and QAC circuits is that AC circuits are allowed fanout ``for free", i.e.\ the input bits to the circuit and the outputs of gates may all be used as inputs to arbitrarily many gates. The quantum analogue of this would be to compute the unitary ``fanout" transformation $U_F$, defined by $U_F \ket{b, x_1, \dotsc, x_{n-1}} = \ket{b, x_1 \oplus b, \dotsc, x_{n-1} \oplus b}$ for $b, x_1, \dotsc, x_{n-1} \in \bits$, or at least to compute this in the case that we call ``restricted fanout" in which $x_1 = \dotsb = x_{n-1} = 0$. \qacz circuits with fanout gates are called \qaczf circuits, and can simulate arbitrary \acz circuits by using ancillae and restricted fanout to make as many copies as needed of the input bits and of the outputs of gates. In fact, \qaczf circuits are \emph{strictly} more powerful than \acz circuits, because \qaczf circuits (even without generalized Toffoli gates) of polynomial size can also compute threshold functions~\cite{HS05,TT16} whereas \acz circuits require exponential size to do so~\cite{Has86}. In contrast, little is known about the power of \qacz circuits and how it compares with that of \acz circuits.

Green et al.~\cite{Gre+02} observed that fanout can be computed by QAC circuits of logarithmic depth and linear size. This raises the question of whether QAC circuits of \emph{sub}logarithmic depth can compute fanout, or at least restricted fanout, even if allowed \emph{arbitrary} size. The same question can be asked about parity, which is a famous example of a function that requires exponential size to compute in \acz~\cite{Has86}, and which is defined for quantum circuits as the unitary transformation $U_\oplus$ such that $U_\oplus \ket{b,x} = \ket{b \oplus \bigoplus_{j=1}^{n-1} x_j, x}$ for $b \in \bits, x = (x_1, \dotsc, x_{n-1}) \in \bits^{n-1}$. In fact, all of these questions are equivalent: Green et al.~\cite{Gre+02} proved that parity and fanout are equivalent up to conjugation by Hadamard gates, and that they reduce to restricted fanout with negligible blowups in size and depth (\cref{rfc}).

\begin{figure*}
	\centerline{
		\begin{quantikz}[row sep = 3.7mm]
			& \gate H & \targ{} \gategroup[4, steps=3, style=densely dashed]{\sc parity} & \targ{} & \targ{} & \gate H & \qw \\
			& \gate H & \ctrl{-1} & \qw & \qw & \gate H & \qw \\
			& \gate H & \qw & \ctrl{-2} & \qw & \gate H & \qw \\
			& \gate H & \qw & \qw & \ctrl{-3} & \gate H & \qw
		\end{quantikz} =
		\begin{quantikz}[row sep = 3.7mm, column sep = 2mm]
			& \gate H & \targ{} & \gate H & \gate H & \targ{} & \gate H & \gate H & \targ{} & \gate H & \qw \\
			& \gate H & \ctrl{-1} & \gate H & \qw & \qw & \qw & \qw & \qw & \qw & \qw \\
			& \qw & \qw & \qw & \gate H & \ctrl{-2} & \gate H & \qw & \qw & \qw & \qw \\
			& \qw & \qw & \qw & \qw & \qw & \qw & \gate H & \ctrl{-3} & \gate H & \qw
		\end{quantikz} =
		\begin{quantikz}[row sep = 3.7mm]
			& \ctrl{1} \gategroup[4, steps=3, style=densely dashed]{\sc fanout} & \ctrl{2} & \ctrl{3} & \ghost H \qw \\
			& \targ{} & \qw & \qw & \ghost H \qw \\
			& \qw & \targ{} & \qw & \ghost H \qw \\
			& \qw & \qw & \targ{} & \ghost H \qw
	\end{quantikz}}
\end{figure*}

\begin{SCfigure}[5]
	\centering
	\begin{quantikz}
		\lstick[wires=3]{\ket x} & \qw & \targ{} &[-5mm] \qw &[-5mm] \qw & \qw & \rstick[wires=3]{$\bigotimes_j \ket{x_j \oplus b}$} \qw \\[-4mm]
		& \qw & \qw & \targ{} & \qw & \qw & \qw \\[-4mm]
		& \qw & \qw & \qw & \targ{} & \qw & \qw \\
		\lstick{\ket b} & \gate[wires=3]{C^{\ }} & \ctrl{-3} & \qw & \qw & \gate[wires=3]{C^\dagger} & \rstick{\ket b} \qw \\
		\lstick[wires=3]{\ket{\zs}} & \qw & \qw & \ctrl{-3} & \qw & & \rstick[wires=3]{\ket{\zs}} \qw \\[-4mm]
		& & \qw & \qw & \ctrl{-3} & & \qw
	\end{quantikz}
	\caption{If $C$ computes restricted fanout, then the bottom circuit computes fanout. (The output label assumes that $\ket{x,b}$ is a standard basis state.)}
	\label{rfc}
\end{SCfigure}
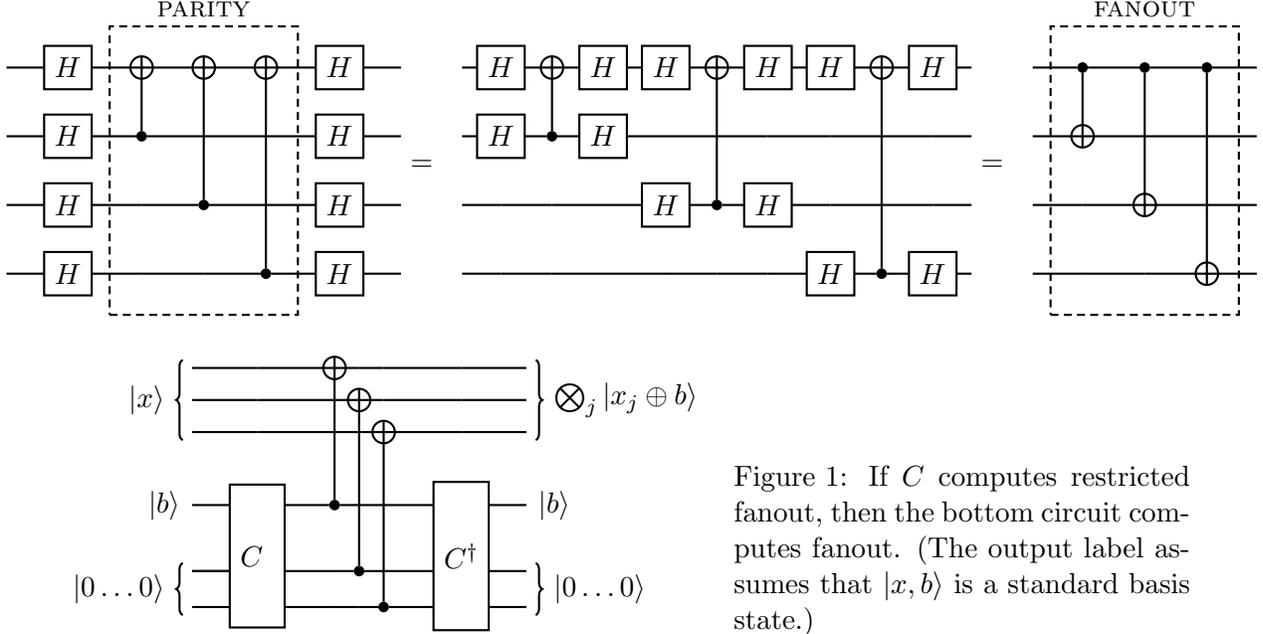

Recent work~\cite{Gok+20,Guo+20} suggests that \qaczf may be a physically realistic model of constant depth computation in certain quantum computing architectures (such as ion traps). As for QAC lower bounds, Fang, Fenner, Green, Homer and Zhang~\cite{Fan+06} proved that QAC circuits with $a$ ancillae require depth at least $\Omega(\log(n/(a+1)))$ to compute the $n$-qubit parity and fanout functions, which is a nontrivial lower bound when $a$ is $o(n)$. Bera~\cite{Ber11} used a different approach to prove something slightly weaker than the $a=0$ case of this result. Finally, Pad\'e, Fenner, Grier and Thierauf~\cite{Pad+20} proved that QAC circuits with two layers of generalized Toffoli gates cannot cleanly\footnote{A clean computation is one in which the ancillae end in the all-zeros state.} compute 4-qubit parity or fanout, regardless of the number of ancillae. A survey of Bera, Green and Homer~\cite{BGH07} discusses some of the aforementioned QAC lower bounds and \qacf upper bounds in greater detail.

\subsection{Results and Selected Proof Overviews}\label{res-intro}
\subsubsection{Definitions of Complexity Measures}

Call $|\ip\dov\aov|^2$ the \emph{fidelity} of states $\ket\aov$ and $\ket\dov$. We define the \emph{size} of a QAC circuit to be the number of multi-qubit gates in it, and the \emph{depth} of a QAC circuit to be the number of layers of multi-qubit gates in it. One motivation for not counting single-qubit gates, besides mathematical convenience, is that size and depth can be interpreted as measures of the reliability and computation time of a quantum circuit respectively, and in practice multi-qubit gates tend to be less reliable and take more time to apply as compared to single-qubit gates.

\subsubsection{Reductions to and from Constructing Nekomata} \label{pep-intro}

Recall that Green et al.~\cite{Gre+02} proved that parity, fanout, and restricted fanout are all equivalent up to low-complexity QAC reductions. In \cref{pep} we make the more general observation that clean approximate and non-clean approximate versions of these problems are all equivalent in this sense. For brevity's sake, here in \cref{res-intro} we will only state immediate corollaries of these reductions insofar as they relate to our other results.

We also introduce another problem equivalent to parity, which all of our results about parity and fanout are proved via reductions to. The state $\frac{1}{\sqrt 2} \sum_{b=0}^1 \ket{b^n}$ is commonly called the \emph{cat state} on $n$ qubits, and we denote it by $\ket{\Cat_n}$. More generally, call a state $\ket\nu$ an \emph{$n$-nekomata} if $\ket\nu = \frac{1}{\sqrt 2} \sum_{b=0}^1 \ket{b^n, \dov_b}$ for some states $\ket{\dov_0}, \ket{\dov_1}$ on any number of qubits (the word ``nekomata" is also the name of two-tailed cats from Chinese and Japanese folklore), or equivalently if a standard-basis measurement of some $n$ qubits of $\ket\nu$ outputs all-zeros and all-ones each with probability 1/2.

Call a QAC circuit $C$ acting on any number of qubits a solution to the ``$p$-approximate $n$-nekomata problem" if there exists an $n$-nekomata $\ket\nu$ such that $C\ket{\zs}$ and $\ket\nu$ have fidelity at least $p$. (There is no need to allow ``ancillae" in this problem, because if $\ket\nu$ is an $n$-nekomata then so is $\ket{\nu,\dov}$ for any state $\ket\dov$.) Note that the identity circuit on $n$ or more qubits trivially solves the 1/2-approximate $n$-nekomata problem. In informal discussions we will often say that a circuit ``constructs an approximate $n$-nekomata" if it solves the $p$-approximate $n$-nekomata problem for some fixed $p \in (1/2,1)$, say $p = 3/4$.

Constructing nekomata reduces to computing restricted fanout because $\ket{\Cat_n} = U_F(H \otimes I)\ket{0^n}$. Our reduction from parity to constructing nekomata is a variant of Green et al.'s~\cite{Gre+02} reduction from parity to restricted fanout.

\subsubsection{Upper Bounds}

\begin{thm}\label{nek-ub-d2}
	For all $\eps>0$ there exists a depth-2 QAC circuit $C$ such that for some $n$-nekomata $\ket\nu$, the fidelity of $C\ket{\zs}$ and $\ket\nu$ is at least $1-\eps$. Furthermore, the size of $C$ and the number of qubits acted on by $C$ are both $\exp(O(n\log(n/\eps)))$.
\end{thm}

To state a stronger upper bound for approximating unitary transformations than can conveniently be done in terms of fidelity, call $1-\norm{\ket\aov - \ket\dov}_2^2$ the \emph{phase-dependent fidelity} of states $\ket\aov$ and $\ket\dov$. This quantity is at most the fidelity of $\ket\aov$ and $\ket\dov$ (\cref{fid-pdfid}).

\begin{cor}\label{par-ub}
	For all $d \ge 7$ and $\eps > 0$ there exist depth-$d$ QAC circuits $C_\oplus, C_F, C_{\sscat}$ of size and number of ancillae $\exp(\poly(n^{1/d}) \log(n/\eps))$, where the $\poly(n^{1/d})$ term is at most $O(n)$, such that for all $n$-qubit states $\ket\iv$,
	\begin{itemize}[leftmargin=*]
		\renewcommand\labelitemi{--}
		\item the phase-dependent fidelity of $C_\oplus \ket{\iv,0\dots 0}$ and $U_\oplus \ket\iv \otimes \ket{\zs}$ is at least $1-\eps$;
		\item the phase-dependent fidelity of $C_F \ket{\iv,\zs}$ and $U_F \ket\iv \otimes \ket{\zs}$ is at least $1-\eps$;
		\item the phase-dependent fidelity of $C_{\sscat}\ket{\zs}$ and $\ket{\Cat_n,\zs}$ is at least $1-\eps$.
	\end{itemize}
\end{cor}

The $d=11$ case of \cref{par-ub} follows immediately from \cref{nek-ub-d2} and our reduction from parity to constructing nekomata. We decrease the minimum depth from 11 to 7 using an optimization specific to the circuit from our proof of \cref{nek-ub-d2}. We prove \cref{par-ub} for higher depths using the fact that $n$-qubit restricted fanout can be computed by a circuit consisting of $d$ layers of $n^{1/d}$-qubit restricted fanout gates.

If we were to also count one-qubit gates toward size and depth, then statements similar to \cref{nek-ub-d2,par-ub} would still hold, because without loss of generality a depth-$d$ QAC circuit acting on $m$ qubits has at most $d+1$ layers of one-qubit gates and at most $(d+1)m$ one-qubit gates.

\subsubsection{Tight Lower Bounds for Constructing Approximate Nekomata in \texorpdfstring{``}{"}Mostly Classical" Circuits} \label{mc-intro}

Call a QAC circuit \emph{mostly classical} if it can be written as $CLM\adj L$ (i.e.\ $C$ is applied last) such that $C$ consists only of generalized Toffoli gates, $L$ is a layer of one-qubit gates, and $M$ is a layer of generalized Toffoli gates. The circuit $C$ here is a close analogue of (classical) AC circuits with bounded fanout, since generalized Toffoli gates can simulate classical AND and NOT gates. The following is apparent from our proof of \cref{nek-ub-d2}:

\begin{rmk}\label{nek-ub-d2-mc}
	\cref{nek-ub-d2} remains true even if ``QAC circuit" is replaced by ``mostly classical QAC circuit".
\end{rmk}

Motivated by \cref{nek-ub-d2-mc}, we prove the following lower bound for constructing approximate nekomata in mostly classical circuits:

\begin{thm}\label{mc-nek-lb-intro}
	Let $C$ be a mostly classical circuit of size $s$ and depth $o(\log n)$, acting on any number of qubits. Then for all $n$-nekomata $\ket\nu$, the fidelity of $C\ket{\zs}$ and $\ket\nu$ is at most
	\begin{equation*}
	\frac{1}2 + 
	\exp\left(-\frac{n^{1-o(1)}}{\max\left(\log s, \sqrt{n}\right)}\right).
	\end{equation*}
\end{thm}

(See \cref{mc-nek-lb} for a more precise tradeoff between depth and fidelity.) In particular, \cref{mc-nek-lb-intro} implies that mostly classical circuits of depth $o(\log n)$ require size at least $\exp(n^{1-o(1)})$ to construct approximate $n$-nekomata, essentially matching the $\exp(\tilde O(n))$ size upper bound from \cref{nek-ub-d2,nek-ub-d2-mc}. This lower bound does not contradict the $\exp(n^{o(1)})$ size upper bounds of depth $\omega(1)$ from \cref{par-ub}, because our reductions between parity, fanout, and constructing nekomata do not in general map mostly classical circuits to mostly classical circuits. Since the identity circuit is mostly classical, the upper bound on the fidelity of $C\ket{\zs}$ and $\ket\nu$ in \cref{mc-nek-lb-intro} is tight up to the value being exponentiated. Finally, if we also allow $r$-qubit parity and fanout gates in mostly classical circuits---a natural model for small values of $r$, in light of the upper bounds from \cref{par-ub}---then a trivial generalization of our proof of \cref{mc-nek-lb-intro} implies that an identical statement holds for circuits of depth $o(\log_{\max(r,2)} n)$.

To prove \cref{mc-nek-lb-intro}, it suffices to prove that the Hamming weight of a standard-basis measurement of any $n$ qubits of $C\ket{\zs}$ is concentrated around some value. We use the fact that standard-basis measurements commute with generalized Toffoli gates, and, after some preparation, apply a concentration inequality of Gavinsky, Lovett, Saks and Srinivasan~\cite{Gav+15}.

\subsubsection{Lower Bounds for Arbitrary QAC Circuits of Low Size and Depth}\label{lbaqcls}

Call the first $n$ qubits of an $n$-nekomata $\frac{1}{\sqrt 2} \sum_{b=0}^1 \ket{b^n, \dov_b}$ the \emph{targets} of that nekomata.

\begin{thm}\label{small-sc}
	There is a universal constant $c>0$ such that the following holds. Let $C$ be a depth-$d$ QAC circuit acting on any number of qubits, and let $\ket\nu$ be an $n$-nekomata such that at most $cn/(d+1)$ multi-qubit gates in $C$ act on the targets of $\ket\nu$. Then the fidelity of $C\ket{\zs}$ and $\ket\nu$ is at most $1/2 + \exp(-\Omega(n/(d+1)))$.
\end{thm}

\begin{cor}\label{small-sc-cor}
	Let $c$ be the constant from \cref{small-sc}. Let $C$ be a depth-$d$ QAC circuit acting on any number of qubits, and assume that, collectively, the first $n$ of these qubits are acted on by at most $cn/(d+1)$ multi-qubit gates in $C$. Then for all states $\ket\dov$,
	\begin{itemize}[leftmargin=*]
		\renewcommand\labelitemi{--}
		\item for $\ket{\iv_\oplus} = \ket{0,+^{n-1}}$, the fidelity of $C\ket{\iv_\oplus, \zs}$ and $U_\oplus\ket{\iv_\oplus} \otimes \ket\dov$ is at most $1/2 + \exp(-\Omega(n/(d+1)))$;
		\item for $\ket{\iv_F} = \ket{+,0^{n-1}}$, the fidelity of $C\ket{\iv_F, \zs}$ and $U_F\ket{\iv_F} \otimes \ket\dov$ is at most $1/2 + \exp(-\Omega(n/(d+1)))$;
		\item the fidelity of $C\ket{\zs}$ and $\ket{\Cat_n,\dov}$ is at most $1/2 + \exp(-\Omega(n/(d+1)))$.
	\end{itemize}
\end{cor}

(Perhaps surprisingly, a sharp ``phase change" near the $cn/(d+1)$ threshold is in fact inherent to our proof. The +1 in $\exp(-\Omega(n/(d+1)))$ is necessary when $C = H$ and $\ket\nu = \ket+ \coloneqq \frac{\ket0 + \ket1}{\sqrt 2}$.) For example, \cref{small-sc} implies that a depth-2 QAC circuit constructing an approximate $n$-nekomata must have at least $\Omega(n)$ multi-qubit gates acting on the targets of that nekomata. This $\Omega(n)$ lower bound is tight, because \cref{nek-ub-d2} says that depth-2 QAC circuits can construct approximate $n$-nekomata, and a depth-$d$ QAC circuit can have at most $nd$ multi-qubit gates acting on any given set of $n$ qubits. Similarly, \cref{small-sc-cor} implies that depth-7 QAC circuits approximating $n$-qubit parity, fanout, or restricted fanout require at least $\Omega(n)$ multi-qubit gates acting on the $n$ ``input" qubits, and this $\Omega(n)$ lower bound is tight as well by \cref{par-ub}.

\cref{small-sc} also implies that the \emph{total} number of multi-qubit gates, a.k.a.\ the size, of a depth-$d$ QAC circuit constructing an approximate $n$-nekomata must be at least $\Omega(n/(d+1))$. When $d$ is $o(\log n)$, this lower bound is disappointingly far from the upper bounds of \cref{nek-ub-d2,par-ub}. However, Green et al.~\cite{Gre+02} observed that for some $d = \Theta(\log n)$, a depth-$d$ QAC circuit of size $O(n)$ can construct an $n$-nekomata (specifically, the $n$-qubit cat state), so for this value of $d$ our $\Omega(n/d)$ size lower bound is tight to within a logarithmic factor. Similarly, for some $d = \Theta(\log n)$, the minimum size of a depth-$d$ QAC circuit that approximates $n$-qubit parity, fanout, or restricted fanout is between $\Omega(n/\log n)$ and $O(n)$, by \cref{small-sc-cor} and upper bounds of Green et al.

If a QAC circuit has size $s \le o(\sqrt n)$ then its depth $d$ satisfies $d \le s \le o(\sqrt n)$, so $s \le o(\sqrt n) \le o(n/(d+1))$. It follows from \cref{small-sc,small-sc-cor} that QAC circuits of \emph{arbitrary} depth require size at least $\Omega(\sqrt n)$ to construct approximate $n$-nekomata, or to approximately compute $n$-qubit parity, fanout, or restricted fanout.\footnote{More generally, inspection of its proof reveals that \cref{small-sc} also holds if ``depth" is replaced by ``maximum number of multi-qubit gates acting on any one of the target qubits". This quantity is at most the \emph{total} number of multi-qubit gates acting on all of the targets, so similar reasoning implies that QAC circuits of arbitrary depth require at least $\Omega(\sqrt n)$ multi-qubit gates \emph{acting on the targets} to construct approximate $n$-nekomata. It follows from our reductions that QAC circuits require at least $\Omega(\sqrt n)$ gates acting on the $n$ ``input" qubits to approximately compute $n$-qubit parity, fanout, or restricted fanout.}

Finally, we remark that \cref{small-sc} is actually a special case of a more general result, \cref{small}, about states $\ket\dov$ such that for some orthogonal projections\footnote{I.e.\ $Q_j = Q_j^2 = \adj Q_j$ for all $j$.} $Q_1, \dotsc, Q_n$ on arbitrary numbers of qubits,
\begin{equation*}
\bra\dov \left(\bigotimes_{j=1}^n Q_j \otimes I\right) \ket\dov =
\bra\dov \left(\bigotimes_{j=1}^n (I -  Q_j)\otimes I\right) \ket\dov = 1/2.
\end{equation*}
(For example, $n$-nekomata satisfy this criterion with $Q_j = \kb0$ for all $j$.) We will comment on this generalization of \cref{small-sc} again in \cref{d2lb-intro}.

\subsubsection{A Normal Form for Quantum Circuits}\label{nf}

Integral to our proof of \cref{small-sc} is a certain normal form for QAC circuits, which may be of independent interest since the standard quantum circuit model is that of QAC circuits whose gates have maximum arity 2. Here we give the underlying intuition, by way of analogy with well-known facts from classical circuit complexity. If we define AC circuits as consisting only of AND and NOT gates, then it cannot in general be assumed that the NOT gates are all adjacent to the inputs. However, by DeMorgan's laws we may equivalently allow OR gates in AC circuits as well, and then it \emph{can} be assumed that the NOT gates are all adjacent to the inputs.\footnote{Invoking this assumption results in a constant-factor blowup in size and no blowup in depth, where (as is customary) we do not count NOT gates toward the size or depth of AC circuits.} Similarly, we introduce a certain further generalization of generalized Toffoli gates which allows us to assume that the one-qubit gates in a QAC circuit are all adjacent to the input.

\subsubsection{Depth-2 Lower Bounds}\label{d2lb-intro}

\begin{thm}\label{d2-lb}
	Let $C$ be a depth-2 QAC circuit of arbitrary size, acting on any number of qubits. Then for all states $\ket\dov$,
	\begin{enumerate}[label=(\roman*), ref={\thethm(\roman*)}]
		\item for $\ket{\iv_\oplus} = \ket{0,+^{n-1}}$, the fidelity of $C\ket{\iv_\oplus, \zs}$ and $U_\oplus\ket{\iv_\oplus} \otimes \ket\dov$ is at most $1/2 + \exp(-\Omega(n))$;\label[thm]{d2-lb-par}
		\item for $\ket{\iv_F} = \ket{+,0^{n-1}}$, the fidelity of $C\ket{\iv_F, \zs}$ and $U_F\ket{\iv_F} \otimes \ket\dov$ is at most $1/2 + \exp(-\Omega(n))$;\label[thm]{d2-lb-fan}
		\item the fidelity of $C\ket{\zs}$ and $\ket{\Cat_n,\dov}$ is at most $1/2 + \exp(-\Omega(n))$.\label[thm]{d2-lb-cat}
	\end{enumerate}
\end{thm}

Our proof of \cref{d2-lb} gives a multiplicative constant of roughly $1/10^{60000}$ implicit in the $\Omega(\cdot)$ notation in the above inequalities, which makes them trivial for small values of $n$. If $n$ is sufficiently large however, then \cref{d2-lb} implies that depth-2 QAC circuits cannot approximate $n$-qubit parity, fanout, or restricted fanout, or approximately construct the $n$-qubit cat state, even if these approximations are not required to be clean. Still taking $n$ to be sufficiently large, this improves on the previously mentioned result of Pad\'e et al.~\cite{Pad+20} that depth-2 QAC circuits cannot cleanly compute parity exactly on four or more qubits.

\cref{d2-lb,par-ub} imply that for all sufficiently large $n$, the minimum depth of a QAC circuit approximating $n$-qubit parity is between 3 and 7 inclusive, and likewise for fanout, restricted fanout, and constructing the cat state. By \cref{nek-ub-d2} there \emph{is} a depth-2 QAC circuit that constructs an approximate $n$-nekomata for all $n$, so any proof of \cref{d2-lb} must use some property of $\ket{\Cat_n,\dov}$ that does not hold for an arbitrary $n$-nekomata. Ours uses a property similar to the fact that if we measure some of the qubits in the ``$B$" register of $\ket{\Cat_n}_A \otimes \ket{\dov}_B$ in an arbitrary basis, then the resulting state in registers $A$ and $B$ is still an $n$-nekomata.

Our proof of \cref{d2-lb} mostly uses different techniques than those of Pad\'e et al. An exception is the observation, of which we use a generalization, that if we define a ``generalized $Z$ gate" on any number of qubits by $Z = I - 2\kb{1\dots1}$ then $Z\ket{0,\iv} = \ket{0,\iv}$ and $Z\ket{1,\iv} = \ket1 \otimes Z\ket\iv$ for all states $\ket\iv$. We also incorporate a variant of the proof given by Bene Watts, Kothari, Schaeffer and Tal~\cite[Theorem 16]{Ben+19} that there is no QNC circuit (QAC circuit whose gates have maximum arity 2) of depth $o(\log n)$ that maps $\ket{\zs}$ to $\ket{\Cat_n, \zs}$: Using a ``light cone" argument they prove that out of any $n$ output qubits, there are at least two whose standard-basis measurements would be independent, but the standard-basis measurements of any two qubits in $\ket{\Cat_n}$ are dependent.

Our proof of \cref{d2-lb} goes roughly as follows. If there are only $o(n)$ multi-qubit gates acting on the $n$ targets of $\ket{\Cat_n,\dov}$ then the result follows from \cref{small-sc}. Otherwise, out of the multi-qubit gates acting on the targets, the average gate acts on $O(1)$ targets, as would be the case in a QNC circuit. Using a variant of a light cone argument, we choose $\Theta(n)$ pairwise disjoint sets of qubits on which to define orthogonal projections, and apply the generalization of \cref{small-sc} that was mentioned at the end of \cref{lbaqcls}.

\subsection{Organization}

In \cref{prelim} we introduce some miscellaneous notation and definitions. In \cref{qac} we give multiple equivalent characterizations of QAC circuits, including the previously mentioned normal form, and introduce some related definitions which we will use in more general contexts as well. In \cref{pep} we give reductions between parity, fanout, restricted fanout, and constructing nekomata; we also use these reductions to prove that the $d \ge 11$ case of \cref{par-ub} follows from \cref{nek-ub-d2}, that \cref{small-sc-cor} follows from \cref{small-sc}, and that \cref{d2-lb-par,d2-lb-fan} follow from \cref{d2-lb-cat}. In \cref{mc-sec} we prove our upper and lower bounds for constructing approximate nekomata in mostly classical circuits, \cref{nek-ub-d2,mc-nek-lb-intro}. In \cref{more-lb} we prove our other main results, \cref{small-sc,d2-lb-cat}. In \cref{d7} we prove the $d<11$ case of \cref{par-ub}. 

Out of \cref{pep,mc-sec,more-lb}, occasionally a later section will reference a self-contained lemma from an earlier section, but otherwise these sections may be read in any order. \cref{d7} relies on content from \cref{pep,mc-sec}.

\subsection{Preliminaries}\label{prelim}

We write $\log$ and $\ln$ to denote the logarithms base 2 and $e$ respectively, and $(x_j)_j$ to denote the tuple of all $x_j$ for $j$ in some implicit index set. Also let $[n] = \{1,\dotsc,n\}$ and $\norm{\psi} = \sqrt{\psi^*\psi}$, i.e.\ $\norm\cdot$ denotes the 2-norm. Anything written as $\bra\cdot$ or $\ket\cdot$ is implicitly unit-length.

\emph{Orthogonal projections} are linear transformations $Q$ such that $Q = Q^2 = \adj Q$. For an orthogonal projection $Q$ and a state $\ket\aov$, we call $\bra\aov Q \ket\aov$ ``the probability that $\ket\aov$ measures to $Q$". If $Q = \kb\dov$ then we also call this ``the probability that $\ket\aov$ measures to $\ket\dov$", and it equals $|\ip\dov\aov|^2$, a.k.a.\ the \emph{fidelity} of $\ket\aov$ and $\ket\dov$. More generally, if $Q$ is an orthogonal projection on some Hilbert space $\Hi$ then we call $\bra\aov (Q \otimes I) \ket\aov$ ``the probability that the $\Hi$ qubits of $\ket\aov$ measure to $Q$".

We use standard notation for the Hadamard basis states $\ket{+} = \frac{\ket 0 + \ket 1}{\sqrt 2}, \ket{-} = \frac{\ket 0 - \ket 1}{\sqrt 2}$, Hadamard gate $H = \ketbra{+}{0} + \ketbra{-}{1}$, and NOT gate $X = \ketbra{0}{1} + \ketbra{1}{0}$. We write $I$ to denote the identity transformation, $I_{\Hi}$ for the identity on the Hilbert space $\Hi$, and $I_n$ for the identity on some $n$-qubit Hilbert space.

To be thorough, we remind the reader that an \emph{$n$-nekomata} is a state with $n$ qubits (called \emph{targets}) that measure to $0^n$ and to $1^n$ each with probability 1/2, or equivalently a state of the form $\frac{1}{\sqrt 2} \sum_{b=0}^1 \ket{b^n, \dov_b}$ for some states $\ket{\dov_0}, \ket{\dov_1}$ on any number of qubits. For example, the $n$-qubit \emph{cat state} is the state $\ket{\Cat_n} = (\ket{0^n} + \ket{1^n})/\sqrt 2$.

\section{QAC Circuits}\label{qac}

Consider a quantum circuit $C$, written as $C = L_d M_d \dotsb L_1 M_1 L_0$ such that each $L_k$ consists only of one-qubit gates and each $M_k$ is a layer (tensor product) of multi-qubit gates. We may assume that each $L_k$ is a single layer as well, because the product of one-qubit gates is also a one-qubit gate. Define the \emph{size} of $C$ to be the number of multi-qubit gates in $C$, the \emph{depth} of $C$ to be the number of layers of multi-qubit gates in $C$ (in this case, $d$), and the \emph{topology} of $C$ to be the set of pairs $(S,k)$ such that $S$ equals the support of some gate in $M_k$, where the \emph{support} of a gate is the set of qubits acted on by that gate. Note that the topology of $C$ encodes its depth, size, and more generally the number of multi-qubit gates acting on any given set of qubits.

Recall that QAC circuits are quantum circuits with arbitrary one-qubit gates and generalized Toffoli gates of arbitrary arity, where $(n+1)$-ary generalized Toffoli gates are defined by $\ket{x,b} \mapsto \ket{x, b \oplus \bigwedge_{j=1}^n x_j}$ for $x = (x_1, \dotsc, x_n) \in \bits^n, b \in \bits$. Define an $(n+1)$-ary OR gate by $\ket{x,b} \mapsto \ket{x, b \oplus \bigvee_j x_j}$ for $x \in \bits^n, b \in \bits$, and call the qubit corresponding to $b$ in these definitions the \emph{target} qubit of the gate. By the construction of an OR gate from a generalized Toffoli gate and NOT gates in \cref{or-gate}, we may add OR gates to the set of allowed gates when defining QAC circuits, without changing the set of topologies of QAC circuits computing any given unitary transformation.

\begin{SCfigure}[50]
	\centering
	\begin{quantikz}[row sep = 2.3mm]
		\ghost X & \orctrl 1 & \qw \\
		\ghost X & \orctrl 1 & \qw \\
		\ghost X & \orctrl 1 & \qw \\
		\ghost X & \gate \vee & \qw
	\end{quantikz} =
	\begin{quantikz}[row sep = 2mm]
		& \gate X & \ctrl 1 & \gate X & \qw \\
		& \gate X & \ctrl 1 & \gate X & \qw \\
		& \gate X & \ctrl 1 & \gate X & \qw \\
		\ghost \vee & \qw & \targ{} & \gate X & \qw
	\end{quantikz}
	\caption{The multi-qubit gates on the left and right are OR and generalized Toffoli gates respectively, whose target qubits are on the bottom wire.}
	\label{or-gate}
\end{SCfigure}
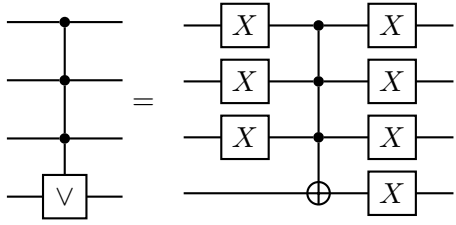

For a state $\ket\cv$ let $\rt{\ket\cv} = \rt\cv = I - 2\kb\cv$ (the R stands for ``reflection"). Let a \emph{mono-product state} be a tensor product of any number of one-qubit states. When $\ket\cv$ is a mono-product state we call $\rt\cv$ an \emph{\rtt gate}. For example, an $(n+1)$-qubit generalized Toffoli gate equals $\rt{\ket{1^n,-}}$, because it acts on the basis $\{\ket0, \ket1\}^{\otimes n} \otimes \{\ket+, \ket-\}$ by multiplying $\ket{1^n,-}$ by -1 and leaving all other states in this basis unchanged.

Consider an $(n+1)$-qubit mono-product state $\ket\cv$, and let $L$ be a layer of one-qubit gates such that $\ket\cv = L\ket{1^n,-}$. Then,
\begin{equation}\label{conj}
\rt\cv = I - 2\kb\cv = I - 2L\kb{1^n,-}\adj L = L(I - 2\kb{1^n,-})\adj L = L\rt{\ket{1^n,-}}\adj L,
\end{equation}
i.e.\ $\rt\cv$ equals the conjugation of a generalized Toffoli gate by a layer of one-qubit gates. (Fang et al.~\cite{Fan+06} observed \cref{conj} in the case where $\ket\cv = \ket{1^{n+1}}$ and $L = I_n \otimes H$.) Therefore, similarly to the above, we may add arbitrary \rtt gates to the set of allowed gates when defining QAC circuits.

In fact, a stronger statement holds. Let a QAC circuit be in \emph{\rtt normal form} if it can be written as $CL$ such that $C$ consists only of multi-qubit \rtt gates and $L$ is a layer of single-qubit gates. We will use the following in \cref{more-lb}:

\begin{prp}\label{rnf}
	Every QAC circuit computes the same unitary transformation as a circuit in \rtt normal form with the same topology.
\end{prp}

\begin{proof}
	The proof is by induction on the depth $d$ of a QAC circuit $C$. If $d=0$ then $C$ is a layer of one-qubit gates, which is already in \rtt normal form. Otherwise write $C = LMD$ such that $L$ is a layer of one-qubit gates, $M$ is a layer of multi-qubit generalized Toffoli gates, and $D$ is a depth-$(d-1)$ QAC circuit. Since $C = LM\adj L L D$, it suffices to prove that $LD$ is equivalent to a circuit in \rtt normal form with the same topology as $D$, and that $LM\adj L$ is equivalent to a layer of \rtt gates that has the same topology as $M$. The first claim follows from the inductive hypothesis. To prove the second claim, note that $LM\adj L = \bigotimes_G L_G G \adj L_G$, where $G$ ranges over all gates in $M$, and $L_G$ is the tensor product of the gates in $L$ that act on the support of $G$. Then apply \cref{conj}.
\end{proof}

\section{Reductions to and from Constructing Nekomata}\label{pep}

The high-level idea of this section may be obtained relatively quickly by inspecting the beginning of \cref{red-sec} (in particular, \cref{pc}) and perhaps also \cref{srrf}, assuming familiarity with certain content from \cref{intro}. Most of the rest of the current section consists of routine calculations.

In \cref{prob-sec} we define the problems mentioned in the following theorem, and in \cref{prob-sec,red-sec} we prove the second and first paragraphs of this theorem respectively:

\begin{thm}\label{red}
	For all $\eps \ge 0$, if there is a QAC circuit of size $s$, depth $d$, and number of qubits acted on $a$ that solves the $(1-\eps)$-approximate $n$-nekomata problem, then there is a QAC circuit of size $O(s+n)$, depth $4d+3$, and number of ancillae $a$ that solves the $(1-O(\eps))$-approximate $(n+1)$-qubit clean parity problem.
	
	For all $0 \le p \le 1$ and every non-red\footnote{Only the arrow from ``nekomata" to ``clean parity" is red.} arrow from a problem P to a problem Q in \cref{flowchart}, if a QAC circuit $C$ solves $p$-approximate, $n$-qubit P then there is a QAC circuit with the same topology as $C$ that solves $p$-approximate, $n$-qubit Q. (If Q is the nekomata problem then substitute ``$n$-nekomata" for ``$n$-qubit nekomata" here.) Furthermore, if this arrow is dashed then $C$ itself solves $p$-approximate, $n$-qubit Q.
\end{thm}

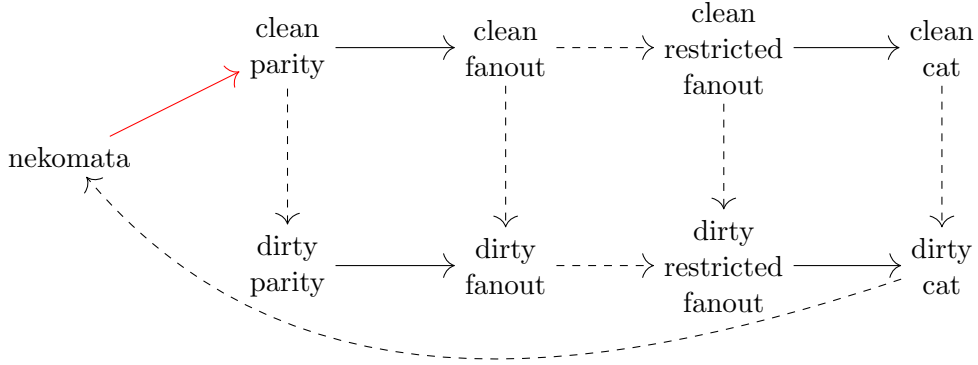
\begin{figure}
	\centering
	\begin{tikzpicture}
		[sty/.style = {align=center},
		arr/.style = {arrows={->[scale=1.8]}},
		arr2/.style = {arrows={[scale=1.8]<->[scale=1.8]}}]
		
		\def \dx {2.9};
		\def \dy {1.45};
		
		\node[sty] (nek) at (0,0) {nekomata};
		\node[sty] (cp) at (\dx,\dy) {clean \\ parity};
		\node[sty] (dp) at (\dx,-\dy) {dirty \\ parity};
		\node[sty] (cf) at (2*\dx,\dy) {clean \\ fanout};
		\node[sty] (df) at (2*\dx,-\dy) {dirty \\ fanout};
		\node[sty] (crf) at (3*\dx,\dy) {clean \\ restricted \\fanout};
		\node[sty] (drf) at (3*\dx,-\dy) {dirty \\ restricted \\ fanout};
		\node[sty] (cc) at (4*\dx,\dy) {clean \\ cat};
		\node[sty] (dc) at (4*\dx,-\dy) {dirty \\ cat};
		
		\draw[arr,red] (nek) to (cp) ;
		\draw[arr] (cp) to (cf);
		\draw[arr,dashed] (cf) to (crf);
		\draw[arr] (crf) to (cc);
		
		\draw[arr] (dp) to (df);
		\draw[arr,dashed] (df) to (drf);
		\draw[arr] (drf) to (dc);
		\draw[arr, dashed,out=198.5,in=310] (dc) to (nek);
		
		\draw[arr,dashed] (cp) to (dp);
		\draw[arr,dashed] (cf) to (df);
		\draw[arr,dashed] (crf) to (drf);
		\draw[arr,dashed] (cc) to (dc);
	\end{tikzpicture}
	\vspace{-9.5mm}
	\caption{A visualization of \cref{red}; see the theorem statement for the meaning of the arrows.}
	\label{flowchart}
\end{figure}

Then, using \cref{red}, in \cref{par-ub-sec} we prove that the $d \ge 11$ case of \cref{par-ub} follows from \cref{nek-ub-d2}. It is easy to prove \cref{small-sc-cor} assuming \cref{small-sc}, and to prove \cref{d2-lb-par,d2-lb-fan} assuming \cref{d2-lb-cat}, using reasoning similar to that in \cref{prob-sec}.
\subsection{Problem Definitions and Most Reductions}\label{prob-sec}

Recall that we define the \emph{phase-dependent fidelity} of states $\ket\aov$ and $\ket\dov$ to be $1 - \norm{\ket\aov - \ket\dov}^2$. This quantity is at most the fidelity of $\ket\aov$ and $\ket\dov$, because
\begin{equation}\label{fid-pdfid}
|\ip\dov\aov|^2
\ge \left(\frac{\ip\dov\aov + \ip\aov\dov}2\right)^2
= \left(1 - \frac{\norm{\ket\aov - \ket\dov}^2}2\right)^2
\ge 1-\norm{\ket\aov - \ket\dov}^2.
\end{equation}
\begin{rmk*}
	If $\ip\aov\dov$ is a real number close to 1, say $\ip\aov\dov = 1-\eps$, then $\ket\aov$ and $\ket\dov$ have fidelity $1-2\eps+\eps^2$ and a nearly identical phase-dependent fidelity of $1-2\eps$. On the other hand, if the phases of $\ket\aov$ and $\ket\dov$ differ, then these states may have low phase-dependent fidelity even if their fidelity is close to 1.
\end{rmk*}

The following two definitions are with respect to an arbitrary unitary transformation $U$ on $n$ qubits:

\begin{prb}[$p$-approximate Clean $U$]\label{clean-unitary}
	Construct a circuit $C$ on at least $n$ qubits such that for all $n$-qubit states $\ket\iv$, the phase-dependent fidelity of $C\ket{\iv, \zs}$ and $U\ket\iv \otimes \ket{\zs}$ is at least $p$.
\end{prb}

\begin{prb}[$p$-approximate Dirty $U$]\label{dirty-unitary}
	Construct a circuit $C$ on at least $n$ qubits such that for all $n$-qubit states $\ket\iv $, the first $n$ qubits of $C\ket{\iv , \zs}$ measure to $U\ket\iv$ with probability at least $p$.
\end{prb}

Any circuit that computes $p$-approximate clean $U$ also computes $p$-approximate dirty $U$, because the probability that the first $n$ qubits of $C\ket{\iv, \zs}$ measure to $U\ket\iv$ is at least the probability that $C\ket{\iv, \zs}$ measures to $U\ket\iv \otimes \ket{\zs}$, a.k.a.\ the fidelity of these two states, which is at least their phase-dependent fidelity.

Given $n$, recall that the unitary transformations for $n$-qubit parity and fanout are defined respectively by $U_\oplus \ket{b,x} = \ket{b \oplus \bigoplus_{j=1}^{n-1} x_j, x}$ and $U_F \ket{b, x} = \ket{b, x_1 \oplus b, \dotsc, x_{n-1} \oplus b}$ for $b \in \bits, x = (x_1, \dotsc, x_{n-1}) \in \bits^{n-1}$. Define the clean and dirty versions of approximating $n$-qubit parity and fanout as instances of \cref{clean-unitary,dirty-unitary} with respect to $U_\oplus$ and $U_F$.

Recall also that $H^{\otimes n} U_\oplus H^{\otimes n} = U_F$~\cite{Gre+02}. We will henceforth write ``$(p,n)$" as an abbreviation for ``$p$-approximate, $n$-qubit". If a circuit $C_\oplus$ computes $(p,n)$ clean parity then the circuit $C_F \coloneqq (H^{\otimes n} \otimes I) C_\oplus (H^{\otimes n} \otimes I)$ computes $(p,n)$ clean fanout, because
\begin{align*}
\max_{\ket\iv} \norm{(C_F - U_F \otimes I)\ket{\iv, \zs}}
&= \max_{\ket\aov} \norm{(H^{\otimes n} \otimes I)(C_F - U_F \otimes I)(H^{\otimes n} \otimes I)\ket{\aov, \zs}} \\
&= \max_{\ket\aov} \norm{(C_\oplus - U_\oplus \otimes I)\ket{\aov, \zs}}.
\end{align*}
Here we made the substitution $\ket\aov = H^{\otimes n}\ket\iv$, and used the facts that $H^2 = I$ and that applying a unitary transformation to a vector does not change the norm of that vector. Similarly, if $C_\oplus$ computes $(p,n)$ \emph{dirty} parity then the same circuit $C_F$ given above computes $(p,n)$ \emph{dirty} fanout, because
\begin{align*}
\min_{\ket\iv} |\bra{\iv,\zs} (\adj U_F \otimes I) C_F \ket{\iv,\zs}|^2 \\
= \min_{\ket\aov} |\bra{\aov,\zs} (H^{\otimes n} \otimes I) (\adj U_F \otimes I) (H^{\otimes n} \otimes I) (H^{\otimes n} \otimes I) C_F (H^{\otimes n} \otimes I) \ket{\aov,\zs}|^2 \\
= \min_{\ket\aov} |\bra{\aov,\zs} (\adj U_\oplus \otimes I) C_\oplus \ket{\aov,\zs}|^2.
\end{align*}

\begin{prb}[$p$-approximate Clean Restricted Fanout]\label{crf}
	Construct a circuit $C$ on at least $n$ qubits such that for all one-qubit states $\ket\iv$, the phase-dependent fidelity of $C\ket{\iv, 0^{n-1}, \zs}$ and $U_F \ket{\iv, 0^{n-1}} \otimes \ket{\zs}$ is at least $p$.
\end{prb}

\begin{prb}[$p$-approximate Dirty Restricted Fanout]\label{drf}
	Construct a circuit $C$ on at least $n$ qubits such that for all one-qubit states $\ket\iv$, the first $n$ qubits of $C\ket{\iv, 0^{n-1}, \zs}$ measure to $U_F \ket{\iv, 0^{n-1}}$ with probability at least $p$.
\end{prb}

Any circuit computing $(p,n)$ clean (resp.\ dirty) fanout trivially computes $(p,n)$ clean (resp.\ dirty) restricted fanout, and similarly to the above, any circuit computing $(p,n)$ clean restricted fanout also computes $(p,n)$ dirty restricted fanout.

\begin{prb}[$p$-approximate Clean $\ket{\Cat_n}$]\label{clean-state}
	Construct a circuit $C$ on at least $n$ qubits such that the phase-dependent fidelity of $C\ket{\zs}$ and $\ket{\Cat_n, \zs}$ is at least $p$.
\end{prb}

\begin{prb}[$p$-approximate Dirty $\ket{\Cat_n}$]\label{dirty-state}
	Construct a circuit $C$ on at least $n$ qubits such that the first $n$ qubits of $C\ket{\zs}$ measure to $\ket{\Cat_n}$ with probability at least $p$.
\end{prb}

By making the substitution $\ket\iv = \ket+$ in \cref{crf,drf} and using the fact that $\ket{\Cat_n} = U_F\ket{+,0^{n-1}}$, it is easy to see that if a circuit $C$ computes $(p,n)$ clean (resp.\ dirty) restricted fanout then the circuit $C(H \otimes I)$ solves the $(p,n)$ clean (resp.\ dirty) cat problem. Similarly to the above, any circuit solving the $(p,n)$ clean cat problem also solves the $(p,n)$ dirty cat problem. Finally, recall the following:

\begin{prb}[$p$-approximate $n$-nekomata]
	Construct a circuit $C$ such that for some $n$-nekomata $\ket \nu$, the fidelity of $C\ket{\zs}$ and $\ket \nu$ is at least $p$.
\end{prb}

Any circuit $C$ solving the $(p,n)$ dirty cat problem also solves the $p$-approximate $n$-nekomata problem, because the probability that the initial qubits of $C\ket{\zs}$ measure to $\ket{\Cat_n}$ equals the maximum over all states $\ket\dov$ of the fidelity of $C\ket{\zs}$ and $\ket{\Cat_n,\dov}$, and the latter state is an $n$-nekomata.

\subsection{Reducing Clean Parity to Constructing Nekomata}\label{red-sec}

\begin{rst*}[first paragraph of \cref{red}]
	For all $\eps \ge 0$, if there is a QAC circuit of size $s$, depth $d$, and number of qubits acted on $a$ that solves the $(1-\eps)$-approximate $n$-nekomata problem, then there is a QAC circuit of size $O(s+n)$, depth $4d+3$, and number of ancillae $a$ that solves the $(1-O(\eps))$-approximate $(n+1)$-qubit clean parity problem.
\end{rst*}

First we reduce exact $(n+1)$-qubit clean parity to the exact $n$-nekomata problem. Let $C$ be a circuit on $a$ qubits such that $\ket\nu \coloneqq C\ket{0^a}$ is an $n$-nekomata. A circuit for exact $(n+1)$-qubit clean parity is shown in \cref{pc}, where the top $n$ wires acted on by each of the $C$ and $\adj C$ subcircuits correspond to the targets of $\ket\nu$. Between times 1 and 2, and also between times 5 and 6, is a layer of $n$ copies of $\rt{\ket{11}}$, the $i$'th of which acts on the wires corresponding to the $i$'th input qubit and the $i$'th target of $\ket\nu$ for $i\in[n]$. The gate $\rt{\ket{11}}$ is better known as a controlled $Z$ gate, and acts as $\ket{xy} \mapsto (-1)^{xy} \ket{xy}$ for $x,y \in \bits$. Between times 3 and 4 is an OR gate (recall \cref{or-gate}).

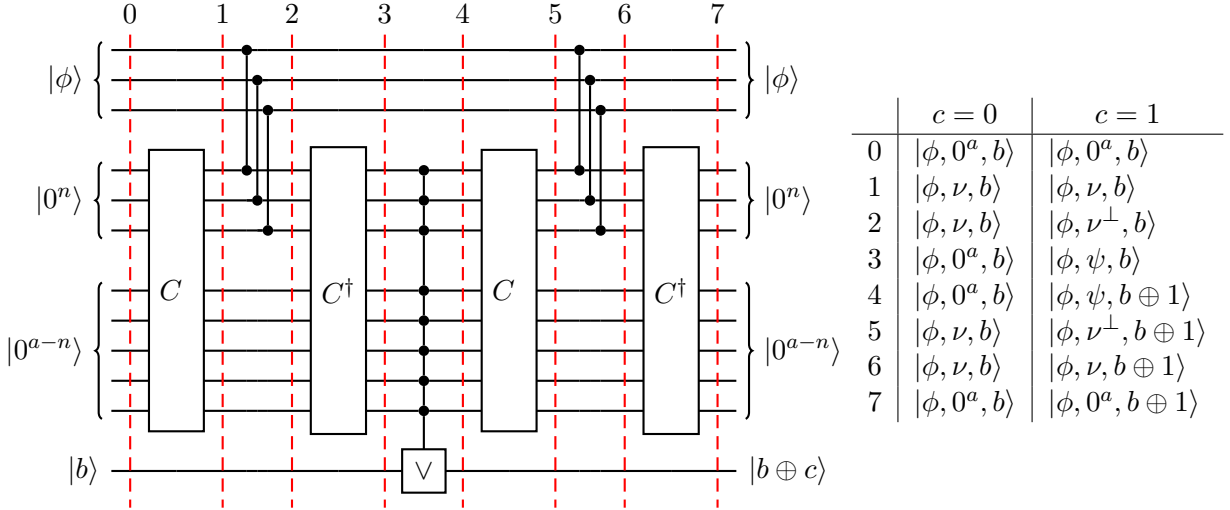
\begin{figure}[!b]
	\centerline{
		\begin{quantikz}[row sep={8mm,between origins}]
			\lstick[wires=3]{\ket\iv} \slice 0 & \qw \slice 1 & \ctrl{3} \qw &[-5mm] \qw &[-5mm] \qw \slice 2 & \qw \slice 3 & \qw \slice 4& \qw \slice 5 & \ctrl{3} \qw &[-5mm] \qw &[-5mm] \qw \slice 6 & \qw \slice 7 &  \rstick[wires=3]{\ket\iv} \qw \\[-4mm]
			& \qw & \qw & \ctrl{3} & \qw \qw & \qw & \qw & \qw & \qw & \ctrl{3} & \qw & \qw & \qw \\[-4mm]
			& \qw & \qw & \qw & \ctrl{3} & \qw & \qw & \qw & \qw & \qw & \ctrl{3} & \qw & \qw \\
			\lstick[wires=3]{\ket{0^n}} & \gate[wires=8]{C^{\ }} & \control{} & \qw & \qw & \gate[wires=8]{\adj C} & \orctrl{1} & \gate[wires=8]{C^{\ }} & \control{} & \qw & \qw & \gate[wires=8]{\adj C} & \rstick[wires=3]{\ket{0^n}} \qw \\[-4mm]
			& & \qw & \control{}\qw & \qw & & \orctrl{1} & & \qw & \control{} & \qw & & \qw \\[-4mm]
			& & \qw & \qw & \control{}\qw & & \orctrl{1}  & & \qw & \qw & \control{} & & \qw \\
			\lstick[wires=5] {\ket{0^{a-n}}} & & \qw & \qw & \qw & & \orctrl{1} & & \qw & \qw & \qw & & \rstick[wires=5] {\ket{0^{a-n}}} \qw \\[-4mm]
			& & \qw & \qw & \qw & & \orctrl{1} & &  \qw & \qw & \qw & & \qw \\[-4mm]
			& & \qw & \qw & \qw & & \orctrl{1} & &  \qw & \qw & \qw & & \qw\\[-4mm]
			& & \qw & \qw & \qw & & \orctrl{1} & &  \qw & \qw & \qw & & \qw \\[-4mm]
			& & \qw & \qw & \qw & & \orctrl{1} & &  \qw & \qw & \qw & & \qw \\
			\lstick{\ket b} & \qw & \qw & \qw & \qw & \qw & \gate \vee &  \qw & \qw & \qw & \qw & \qw & \rstick{\ket{b\oplus c}} \qw
		\end{quantikz}
		\begin{tabular}{l|l|l}
			& \multicolumn{1}{c|}{$c = 0$} & \multicolumn{1}{c}{$c=1$} \\ \hline
			0 & $\ket{\iv,0^a,b}$ & $\ket{\iv,0^a,b}$ \\
			1 & $\ket{\iv,\nu,b}$ & $\ket{\iv,\nu,b}$ \\
			2 & $\ket{\iv,\nu,b}$ & $\ket{\iv,\nu^\perp,b}$ \\
			3 & $\ket{\iv,0^a,b}$ & $\ket{\iv,\dov,b}$ \\
			4 & $\ket{\iv,0^a,b}$ & $\ket{\iv,\dov,b\oplus 1}$ \\
			5 & $\ket{\iv,\nu,b}$ & $\ket{\iv,\nu^\perp,b\oplus 1}$ \\
			6 & $\ket{\iv,\nu,b}$ & $\ket{\iv,\nu,b\oplus 1}$ \\
			7 & $\ket{\iv,0^a,b}$ & $\ket{\iv,0^a,b\oplus 1}$
	\end{tabular}}
	\caption{A circuit for parity, assuming $C$ constructs an $n$-nekomata. See the surrounding text for information about the other variables mentioned.}
	\label{pc}
\end{figure}

Let $\ket{\nu^\perp} = ((\kb{0^n} - \kb{1^n}) \otimes I_{a-n})\ket\nu$ and $\ket\dov = \adj C \ket{\nu^\perp}$. Then $\ip{0^a}\dov = \bra{0^a} \adj C \ket{\nu^\perp} = \ip\nu{\nu^\perp} = 1/2 - 1/2 = 0$, so $\ket\dov$ is a superposition of standard basis states with Hamming weight at least 1. It follows that if $\ket\iv$ is an $n$-qubit standard basis state with parity $c$, and if $b$ is a classical bit, then the table in \cref{pc} indicates the state at each time in the computation. In particular, the circuit in \cref{pc} correctly computes parity on the input $\ket{\iv,b}$, so by linearity this circuit correctly computes parity on arbitrary inputs.

\begin{rmk*}
	The truncation of the circuit in \cref{pc} at time 4 does not compute dirty parity, because if $\ket\iv$ is a superposition of standard basis states with different parities then at time 4 the ancillae are entangled with the input qubits.
\end{rmk*}

Now assume only that there exists an $n$-nekomata $\ket\nu$ such that $|\bra\nu C \ket{0^a}|^2 \ge 1-\eps$. Without loss of generality, multiply $\ket\nu$ by a phase factor so that $\bra\nu C \ket{0^a}$ is a nonnegative real number, i.e.\ $\bra \nu C \ket{0^a} \ge \sqrt{1-\eps} \ge 1-\eps$. To prove that the circuit $D$ from \cref{pc} computes $(1-O(\eps))$-approximate clean parity, it suffices to show that $\norm{(D - U_\oplus \otimes I_a)\ket{\Phi,0^a}} \le O(\sqrt\eps)$ for all $(n+1)$-qubit states $\ket\Phi$. Write $\ket\Phi = \sum_{b,c=0}^1 \alpha_{b,c} \ket{b,\iv_{b,c}}$ such that $\sum_{b,c} |\alpha_{b,c}|^2 = 1$, and $\ket{\iv_{b,c}}$ is a superposition of $n$-qubit standard basis states with parity $c$. By the triangle inequality it suffices to prove that $\norm{(D - U_\oplus \otimes I)\ket{\iv_{b,c},0^a,b}} \le O(\sqrt\eps)$ for all $b,c$, where we have written the expressions $\iv_{b,c},0^a,b$ in the same order as in \cref{pc}.

Again let $\ket{\nu^\perp} = ((\kb{0^n} - \kb{1^n}) \otimes I_{a-n})\ket\nu$, and now let $\ket\dov = \frac{(I - \kb{0^a}) \adj C \ket{\nu^\perp}} {\norm{(I - \kb{0^a}) \adj C \ket{\nu^\perp}}}$. Since $\ip\nu{\nu^\perp} = 0$ and $\ip{0^a}\dov = 0$, there exists an $a$-qubit unitary transformation $U$ such that $U\ket{0^a} = \ket\nu$ and $U\ket\dov = \ket{\nu^\perp}$. For $t \in [7]$ let $L_t$ be the subcircuit between times $t-1$ and $t$ in \cref{pc}, and let $M_t$ be the circuit formed by substituting $U$ for $C$ and $\adj U$ for $\adj C$ in $L_t$. Fix some $b$ and $c$, and let $\ket\iv = \ket{\iv_{b,c}}$. Then, by the triangle inequality,
\begin{align*}
	\norm{(D - U_\oplus \otimes I)\ket{\iv,0^a,b}}
	&= \norm{(L_7 \dotsb L_1 - M_7 \dotsb M_1) \ket{\iv, 0^a, b}} \\
	&\le \sum_{t=1}^7 \norm{L_7 \dotsb L_{t+1} (L_t - M_t) M_{t-1} \dotsb M_1 \ket{\iv, 0^a, b}} \\
	&= \sum_{t=1}^7 \norm{(L_t - M_t) M_{t-1} \dotsb M_1 \ket{\iv, 0^a, b}}
	\coloneqq \delta.
\end{align*}
If $c = 0$ then
\begin{align}\label{cz}
\begin{split}
\delta
	&= 2\norm{(C - U)\ket{0^a}} + 2\norm{(\adj C - \adj U)\ket\nu}
	= 2\norm{C\ket{0^a} - \ket\nu} + 2\norm{C(\adj C - \adj U)\ket\nu} \\
	&= 4\norm{C\ket{0^a} - \ket\nu}
	= 4\sqrt{2 - 2\bra\nu C \ket{0^a}}
	\le 4\sqrt{2-2(1-\eps)}
	\le O(\sqrt\eps).
\end{split}
\end{align}
If $c=1$ then
\begin{equation*}
\delta = \norm{(C-U) \ket{0^a}} + \norm{(\adj C - \adj U) \ket{\nu^\perp}} + \norm{(C-U) \ket\dov} + \norm{(\adj C - \adj U) \ket\nu},
\end{equation*}
and it follows from \cref{cz} that
\begin{equation*}
\delta
= O(\sqrt\eps) + \norm{C(\adj C - \adj U) \ket{\nu^\perp}} + \norm{C\ket\dov - \ket{\nu^\perp}}
= O(\sqrt\eps) + 2\norm{C\ket\dov - \ket{\nu^\perp}}.
\end{equation*}
Finally, since $(I - \kb\nu) - \kb{\nu^\perp}$ is positive semidefinite,
\begin{align*}
\norm{C\ket\dov - \ket{\nu^\perp}}^2
&= 2 - \bra{\nu^\perp}C\ket{\dov} - \bra\dov \adj C \ket{\nu^\perp}
= 2 - 2\norm{(I - \kb{0^a}) \adj C \ket{\nu^\perp}} \\
&= 2 - 2\sqrt{1 - |\bra{\nu^\perp} C \ket{0^a}|^2}
\le 2 - 2\sqrt{1 - \norm{(I - \kb\nu) C \ket{0^a}}^2} \\
&= 2 - 2\bra\nu C \ket{0^a} \le 2 - 2(1-\eps) \le O(\eps).
\end{align*}

\subsection{Proof of \texorpdfstring{\cref{par-ub} ($d \ge 11$)}{Corollary 1.2 (d >= 11)} Assuming \texorpdfstring{\cref{nek-ub-d2}}{Theorem 1.1}}\label{par-ub-sec}

\begin{lem}[essentially Green et al.~\cite{Gre+02}]\label{srrf}
	For all $m \ge 2$ there is a quantum circuit of depth $\ceil{\log_m n}$ and size at most $n-1$, consisting only of restricted fanout gates of arity at most $m$, that computes $n$-qubit restricted fanout exactly using no ancillae.
\end{lem}

\begin{proof}
	By linearity it suffices to consider input states of the form $\ket{b, 0^{n-1}}$ for $b \in \bits$. The proof is by induction on $d = \ceil{\log_m n}$, for a fixed value of $m$. Note that $d-1 < \log_m n \le d$, so $m^{d-1} < n \le m^d$. If $d = 0$ then $n=1$ and the identity circuit suffices. If $d>0$ then by induction we can map $\ket{b,0^{n-1}}$ to $\ket{b^{m^{d-1}}, 0^{n - m^{d-1}}}$ in depth $d-1$ and size at most $m^{d-1}-1$. Let $n_1, \dotsc, n_{m^{d-1}} \in [m]$ be such that $\sum_i n_i = n$, and compute $\bigotimes_i U_F \ket{b, 0^{n_i - 1}} = \ket{b^n}$. Since $\bigotimes_i U_F$ has size at most $n-m^{d-1}$ (omitting one-qubit gates), the total size of the circuit is at most $(n-m^{d-1}) + (m^{d-1}-1) = n-1$.
\end{proof}

\begin{rst*}[\cref{par-ub}]
	For all $d \ge 7$ and $\eps > 0$ there exist depth-$d$ QAC circuits $C_\oplus, C_F, C_{\sscat}$ of size and number of ancillae $\exp(\poly(n^{1/d}) \log(n/\eps))$, where the $\poly(n^{1/d})$ term is at most $O(n)$, such that for all $n$-qubit states $\ket\iv$,
	\begin{itemize}[leftmargin=*]
		\renewcommand\labelitemi{--}
		\item the phase-dependent fidelity of $C_\oplus \ket{\iv,0\dots 0}$ and $U_\oplus \ket\iv \otimes \ket{\zs}$ is at least $1-\eps$;
		\item the phase-dependent fidelity of $C_F \ket{\iv,\zs}$ and $U_F \ket\iv \otimes \ket{\zs}$ is at least $1-\eps$;
		\item the phase-dependent fidelity of $C_{\sscat}\ket{\zs}$ and $\ket{\Cat_n,\zs}$ is at least $1-\eps$.
	\end{itemize}
\end{rst*}

\begin{proof}[Proof ($d \ge 11$) assuming \cref{nek-ub-d2}]
	Recall that \cref{nek-ub-d2} states that for all $\eps>0$, the $(1-\eps)$-approximate $n$-nekomata problem is solvable by a depth-2 QAC circuit of size and number of qubits acted on $\exp(O(n\log(n/\eps)))$. It follows from \cref{red} that for all $\eps>0$, the $(1-\eps)$-approximate $n$-qubit clean parity, fanout, and cat problems are solvable by depth-11 QAC circuits of size and number of ancillae $\exp(O(n\log(n/\eps)))$.
	
	Now let $\eps>0$ and $N,d \in \N$, and let $n = \ceil{N^{1/d}}$. By \cref{srrf} there exists a depth-$d$ circuit of size at most $N-1$, consisting only of $(\le n)$-ary restricted fanout gates, that computes $N$-qubit restricted fanout exactly using no ancillae. Replace each gate in this circuit with the aforementioned depth-11 QAC circuit computing $(1-\eps)$-approximate clean restricted fanout on the appropriate number of qubits. The result is a depth-$11d$ QAC circuit of size and number of ancillae $(N-1) \cdot \exp(O(n\log(n/\eps)))$, and by an argument involving the triangle inequality similar to that in \cref{red-sec}, it computes $(1-O(N^2\eps))$-approximate $N$-qubit clean restricted fanout. It follows from \cref{red} that there exist depth-$O(d)$ QAC circuits of size and number of ancillae $O(N) \cdot \exp(O(n\log(n/\eps)))$ that compute $(1-O(N^2\eps))$-approximate $N$-qubit clean parity, fanout and cat. Finally, given $\mathcal E> 0$ and a sufficiently large value $D$, substituting appropriate values $\eps = \Theta(\mathcal E/N^2)$ and $d = \Theta(D)$ into the above gives depth-$D$ QAC circuits solving the $(1-\mathcal E)$-approximate $N$-qubit clean parity, fanout, and cat problems, where the size and number of ancillae are $\exp(\poly(N^{1/D}) \log(N/\mathcal E))$ (using the fact that $n \le N^{1/d}+1 = \poly(N^{1/d})$).
\end{proof}

\section{Tight Bounds for Constructing Approximate Nekomata in \texorpdfstring{``}{"}Mostly Classical" Circuits} \label{mc-sec}
Call a QAC circuit \emph{purely classical} if it consists only of generalized Toffoli gates (including NOT gates, which are generalized Toffoli gates on one qubit). Call a QAC circuit \emph{mostly classical} if it can be written as $CL$ such that $C$ is purely classical and $L$ is a layer of \rtt gates; by \cref{conj} this is equivalent to the definition from \cref{mc-intro}. Call a mostly classical QAC circuit \emph{nice} if it can be written as $CL$ in this way such that every multi-qubit gate $\rt\cv$ in $L$ satisfies $|\ip{\zs}\cv|^2 \le 1/4$. (The niceness condition will allow us to express certain quantities as convex combinations in a convenient way, by ensuring that the coefficients in these convex combinations are between 0 and 1.) We prove the following generalizations of \cref{nek-ub-d2,mc-nek-lb-intro} respectively:

\begin{thm}\label{nek-ub}
	For all $2 \le d \le \log n$ and $\eps>0$ there exists a nice, mostly classical, depth-$d$ QAC circuit $C$ of size and number of qubits acted on $\exp(O(n2^{-d} \log(n2^{-d}/\eps))) + O(n)$ such that $C\ket\zs$ has fidelity at least $1-\eps$ with some $n$-nekomata.
\end{thm}

\begin{thm}\label{mc-nek-lb}
	Let $C$ be a mostly classical circuit of size $s$ and depth $d$.
	\begin{enumerate}[label=(\roman*), ref={\thethm(\roman*)}]
		\item The fidelity of $C\ket{\zs}$ and any $n$-nekomata is at most\label[thm]{mc-nek-lb-mean}
		\begin{equation*}
			\frac{1}2 + 
			\exp\left(-\Omega\left(\frac{n/(4^d \log n)}{\max\left(\log s, \sqrt{n/(4^d \log n)}\right)}\right)\right).
		\end{equation*}
		\item If $C$ is nice, then the fidelity of $C\ket{\zs}$ and any $n$-nekomata is at most\label[thm]{mc-nek-lb-nice}
		\begin{equation*}
		\frac{1}2 + \exp\left(-\Omega\left(\frac{n/2^d}{\max\left(\log s, \sqrt{n/2^d} \right)} \right) \right).
		\end{equation*}
	\end{enumerate}
\end{thm}

\cref{nek-ub,mc-nek-lb-nice} imply that for $d \ge 2$, the minimum size of a nice, mostly classical, depth-$d$ QAC circuit that ``constructs an approximate $n$-nekomata" (i.e.\ maps $\ket{\zs}$ to a state that has fidelity at least 3/4 with some $n$-nekomata) is between $\exp(\Omega(n/2^d))$ and $\exp(\tilde O(n/2^d)) + O(n)$. We prove the $d>2$ case of \cref{nek-ub} solely for the sake of comparison with \cref{mc-nek-lb-nice}, as \cref{nek-ub} gives a weaker upper bound than \cref{par-ub} when $\omega(1) \le d \le o(\log n)$. \cref{mc-nek-lb} makes a stronger statement about nice circuits than about non-nice circuits, since $a/\max(\log s, \sqrt a) = \min(a/\log s, \sqrt a)$ for all $a>0$.

In \cref{samp-sec} we make some general observations about mostly classical circuits and ``approximate nekomata", including observations common to the proofs of \cref{nek-ub,mc-nek-lb}. In \cref{mc-ub-sec} we prove \cref{nek-ub}, and in \cref{mc-lb-sec} we prove \cref{mc-nek-lb-nice}. We prove \cref{mc-nek-lb-mean} in \cref{mc-app}, because its proof has a similar high-level idea to that of \cref{mc-nek-lb-nice} and is much more complicated.
\subsection{Reduction to a Classical Sampling Problem}\label{samp-sec}

Collectively, the following observations reduce proving \cref{nek-ub,mc-nek-lb} to proving upper and lower bounds respectively for a certain type of sampling problem. This sampling problem can be succinctly characterized in purely classical and probabilistic terms, with only a transient reference to quantum circuits.

Recall that nekomata can be defined as states for which a standard-basis measurement of the targets is distributed in a certain way. The following two lemmas make similar statements about ``approximate nekomata", and are used to prove \cref{nek-ub,mc-nek-lb} respectively:

\begin{lem}\label{eg-eps-nek}
	Let $\ket\aov$ be a state with $n$ ``target" qubits that measure to all-zeros with probability exactly $1/2$ and all-ones with probability at least $1/2 - (2/3)\eps$. Then there exists an $n$-nekomata $\ket\nu$ such that $|\ip\nu\aov|^2 \ge 1-\eps$.
\end{lem}
\begin{proof}
	Let $\delta \le (2/3)\eps$ be such that the targets of $\ket\aov$ measure to all-ones with probability $1/2 - \delta$. Let $\ket\nu = \frac{1}{\sqrt 2} \sum_{b=0}^1 \frac{(\kb{b^n} \otimes I) \ket\aov} {\norm{(\kb{b^n} \otimes I) \ket\aov}}$, where $\kb{b^n}$ acts on the targets of $\ket\aov$, and note that $\ket\nu$ is an $n$-nekomata. Then,
	\begin{align*}
	\ip\aov\nu^2 &= \left(\frac{1}{\sqrt 2}\sum_{b=0}^1 \norm{(\kb{b^n} \otimes I) \ket\aov}\right)^2
	= \frac{1}{2}\left(\sqrt{\frac{1}{2}} + \sqrt{\frac{1}{2} - \delta}\right)^2
	= \frac{1}{2}\left(1-\delta + \sqrt{1-2\delta}\right) \\
	&\ge \frac{1}{2}(1-\delta+1-2\delta) = 1 - (3/2)\delta \ge 1-\eps. \qedhere
	\end{align*}
\end{proof}

\begin{lem}\label{nek-lb}
	Let $\ket\aov$ be a state with $n$ ``target" qubits that measure to all-zeros with probability $p$ and all-ones with probability $q$. Then $|\ip \nu \aov|^2 \le 1/2 + \sqrt{\min(p,q)}$ for all $n$-nekomata $\ket \nu$ with the same targets as $\ket\aov$.
\end{lem}
\begin{proof}
	Let $Q_b = \kb{b^n} \otimes I$ for $b \in \bits$. By the triangle inequality and Cauchy-Schwarz,
	\begin{equation*}
	|\ip\nu\aov|
	= |\bra\nu(Q_0 + Q_1)\ket\aov|
	\le \sum_{b=0}^1 |\bra\nu Q_b \ket\aov|
	\le \sum_{b=0}^1 \norm{Q_b \ket\nu} \cdot \norm{Q_b \ket\aov}
	= \sqrt{p/2} + \sqrt{q/2},
	\end{equation*}
	so $|\ip\nu\aov|^2 \le p/2 + q/2 + \sqrt{pq} \le 1/2 + \sqrt{\min(p,q)}$.
\end{proof}

Consider a mostly classical circuit, written as $CL$ such that $C$ is purely classical and $L$ is a layer of \rtt gates. A standard-basis measurement of designated ``target" qubits of $CL\ket{\zs}$ is distributed identically to an appropriate marginal distribution of a standard-basis measurement of \emph{all} qubits of $CL\ket{\zs}$. It is easy to see that standard-basis measurements commute with generalized Toffoli gates, so we may first measure $L\ket{\zs}$ in the standard basis and then apply $C$ to the result.

Finally, the following is straightforward to verify:

\begin{lem}\label{d1-distr}
	Let $(\ket{\cv_j})_j$ be one-qubit states, and let $p_j = |\ip1{\cv_j}|^2$ for all $j$. A standard-basis measurement of $\rt{\bigotimes_j \ket{\cv_j}} \ket{\zs}$ outputs all-zeros with probability $\left(1 - 2\prod_j (1-p_j) \right)^2$, and any other boolean string $(y_j)_j$ with probability $4\prod_j  (1-p_j) P(\bern(p_j) = y_j)$.
\end{lem}
\begin{proof}
	Let $z=\ket\zs$, let $\ket\cv = \bigotimes_j \ket{\cv_j}$, and recall that $\rt\cv = I-2\kb\cv$. Clearly,
	\begin{equation*}
		|\bra{z} \rt\cv \ket{z}|^2 =
		|\ip{z}{z} - 2\ip{z}\cv\ip\cv{z}|^2 =
		\left|1-2\prod_j|\ip0{\cv_j}|^2\right|^2 =
		\left(1-2\prod_j(1-p_j)\right)^2.
	\end{equation*}
	Similarly, if $y=(y_j)_j$ is any boolean string besides the all-zeros string, then
	\begin{align*}
		|\bra{y} \rt\cv \ket{z}|^2
		&= |\ip{y}z - 2\ip{y}\cv\ip\cv{z}|^2
		= \left|-2\prod_j \ip{y_j}{\cv_j} \ip{\cv_j}0\right|^2
		= 4\prod_j|\ip{y_j}{\cv_j}|^2|\ip0{\cv_j}|^2 \\
		&= 4\prod_j  (1-p_j) P(\bern(p_j) = y_j). \qedhere
	\end{align*}
\end{proof}

For mostly classical circuits that are nice, the following is a more convenient characterization of this distribution:

\begin{cor}\label{convex-comb}
	If $\prod_j (1-p_j) \le 1/4$ then the distribution from \cref{d1-distr} is a convex combination of all-zeros with probability $1 - 4\prod_j (1-p_j)$ and $(\bern(p_j))_j$ with probability $4\prod_j (1-p_j)$, where the $\bern(p_j)$ random variables are all independent.
\end{cor}
\begin{proof}
	$\left(1 - 2\prod_j (1-p_j) \right)^2 = \left(1 - 4\prod_j(1-p_j)\right) + 4\prod_j(1-p_j)^2$, and $4\prod_j(1-p_j)^2 = 4\prod_j(1-p_j)P(\bern(p_j)=0)$.
\end{proof}

\subsection{Proof of \texorpdfstring{\cref{nek-ub}}{Theorem 4.1}}\label{mc-ub-sec}

We first prove the depth-2 case of \cref{nek-ub}, and then we reduce the general case to the depth-2 case.
\begin{rst*}[depth-2 case of \cref{nek-ub}]
	For all $\eps>0$ there exists a nice, mostly classical, depth-2 QAC circuit $C$ of size and number of qubits acted on $\exp(O(n\log(n/\eps)))$ such that $C\ket{\zs}$ has fidelity at least $1-\eps$ with some $n$-nekomata.
\end{rst*}

\begin{proof}
	Let $M \in \N$ and $\delta \in (0,1)$ be parameters to be chosen later.\footnote{Ultimately we will let $M = \exp(\Theta(n\log(n/\eps)))$ and $\delta = \Theta(\eps/n)$.} The circuit acts on $n(M+1)$ qubits, all initialized to $\ket 0$, and arranged in a grid of dimensions $n \times (M+1)$ (\cref{grid}). Designate one column as the ``target" column, and call the qubits in the $M$ other columns ``ancillae". First, to each ancilla column, apply $\rt{\left(\sqrt \delta \ket 0 + \sqrt{1-\delta} \ket 1\right)^{\otimes n}}$. Second, to each row, apply an $(M+1)$-qubit OR gate whose target qubit is in the target column. (A layer of OR gates is a depth-1 purely classical circuit, by the construction in \cref{or-gate}.)
	
	All measurements described below are with respect to the state on the ancillae between the first and second layers of the above circuit. By \cref{eg-eps-nek} it suffices to choose $M$ and $\delta$ such that if we measure the ancillae in the standard basis, then with probability exactly 1/2 all of the ancillae measure to 0, and with probability at least $1/2 - (2/3)\eps$ at least one ancilla in each row measures to 1. We now choose $\delta$ in terms of $M$ such that the ancillae measure to all-zeros with probability 1/2. By \cref{d1-distr} and the independence of measurements of different columns, it suffices to ensure that $(1 - 2\delta^n)^{2M} = 1/2$. Choose $\delta \in (0, (1/2)^{1/n})$ that satisfies this equation.
	
	\begin{figure}
		\centering
		\begin{tikzpicture}
			\draw[step=.5cm] (0,0) grid (11,-2);
			\draw[decorate, decoration = {brace, amplitude=7pt, raise=4pt}] (0,-2) -- (0,0) node[midway, left = 10pt] {$n$};
			\draw[decorate, decoration = {brace, amplitude=7pt, raise=4pt}] (0,0) -- (11,0) node[midway, yshift = 17pt] {$M+1$};
		\end{tikzpicture}
		\caption{}\label{grid}
	\end{figure}
	
	Let $\eps^\prime = (2/3)\eps$. Below we will choose $M$ such that the probability that there exists an ancilla column measuring to neither all-zeros nor all-ones is at most $\eps^\prime$. Equivalently, with probability at least $1 - \eps^\prime$, every ancilla column measures to either all-zeros or all-ones. Since the ancillae measure to all-zeros with probability 1/2, it follows that with probability at least $1/2 - \eps^\prime$, every ancilla column measures to all-zeros or all-ones \emph{and} at least one ancilla column measures to all-ones. Therefore the probability is at least $1/2-\eps^\prime$ that at least one ancilla in every row measures to 1, as desired.
	
	By \cref{d1-distr} and a union bound, the probability that there  exists an ancilla column measuring to neither all-zeros nor all-ones is at most
	\begin{equation*}
	M(1 - (1-2\delta^n)^2 - 4\delta^n(1-\delta)^n) = 4M\delta^n(1 - \delta^n - (1-\delta)^n) \le 4Mn\delta^{n+1}.
	\end{equation*}
	Since $1/2 = (1-2\delta^n)^{2M} \le \exp(-4\delta^n M)$, it follows that $\delta^n \le \ln(2)/4M$, so
	\begin{equation*}
	4Mn\delta^{n+1} \le 4Mn(\ln(2)/4M)^{1+1/n} = \ln(2) n (\ln(2)/4M)^{1/n}.
	\end{equation*}
	To make this bound at most $\eps^\prime$, let
	\begin{equation*}
	M = \ceil{(\ln(2)/4) \cdot (\ln(2)n/\eps^\prime)^n} \le \exp(O(n\log(n/\eps))).
	\end{equation*}
	Finally, the circuit is nice because $\delta^n \le \ln(2)/4M \le \ln(2)/4 < 1/4$.
\end{proof}

\begin{rst*}[\cref{nek-ub}]
	For all $2 \le d \le \log n$ and $\eps>0$ there exists a nice, mostly classical, depth-$d$ QAC circuit $C$ of size and number of qubits acted on $\exp(O(n2^{-d} \log(n2^{-d}/\eps))) + O(n)$ such that $C\ket\zs$ has fidelity at least $1-\eps$ with some $n$-nekomata.
\end{rst*}

\begin{proof}
	Since a two-qubit generalized Toffoli gate (a.k.a.\ a C-Not gate) computes two-qubit restricted fanout,\footnote{\vspace{-\baselineskip}Recall that a circuit $U_F$ computes $n$-qubit restricted fanout if $U_F\ket{b,0^{n-1}} = \ket{b^n}$ for all $b \in \bits$.} it follows from \cref{srrf} that $N$-qubit restricted fanout is computable by a purely classical circuit of depth $\ceil{\log N}$, size $N-1$, and no ancillae. Use the circuit from the depth-2 case of \cref{nek-ub} to construct a state $\ket\aov$ that has fidelity at least $1-\eps$ with some $m$-nekomata for $m = \ceil{n/2^{d-2}}$, and initialize $n-m$ additional qubits to $\ket{0^{n-m}}$. Partition the $m$ targets of $\ket\aov$ (more precisely, of the $m$-nekomata approximated by $\ket\aov$) and the $n-m$ new qubits into $m$ sets, each of size at most $\ceil{n/m} \le 2^{d-2}$, and each including one of the targets of $\ket\aov$. To each of these sets of qubits, apply the depth-$(\le d-2)$ circuit for restricted fanout described above, where the qubit being ``fanned out" is the target of $\ket\aov$ in that set.	
\end{proof}

\subsection{Proof of \texorpdfstring{\cref{mc-nek-lb-nice}}{Theorem 4.2(ii)}}\label{mc-lb-sec}

We use the following concentration inequality of Gavinsky, Lovett, Saks and Srinivasan~\cite{Gav+15}:

\begin{dfn}[{\cite{Gav+15}}]
	Call a random string $(Y_1, \dotsc, Y_n) \in \bits^n$ a \emph{read-$r$ family} if there exist $m \in \N$, independent random variables $X_1, \dotsc, X_m$, sets $S_1, \dotsc, S_n \subseteq [m]$ such that $|\{j \mid i \in S_j\}| \le r$ for all $i \in [m]$, and functions $f_1, \dotsc, f_n$ such that $Y_j = f_j((X_i)_{i\in S_j})$ for all $j \in [n]$.
\end{dfn}

\begin{thm}[{\cite{Gav+15}}]\label{glss}
	Let $(Y_1, \dotsc, Y_n)$ be a read-$r$ family, and let $\mu = \Ex\left[\sum_{j=1}^n Y_j\right]$. Then for all $\eps\ge0$,
	\begin{align*}
	P(Y_1 + \dotsb + Y_n \ge \mu + \eps n) &\le \exp(-2\eps^2 n/r), \\
	P(Y_1 + \dotsb + Y_n \le \mu - \eps n) &\le \exp(-2\eps^2 n/r).
	\end{align*}
\end{thm}

\begin{rmk*}
	For example, if $r=1$ then $Y_1, \dotsc, Y_n$ are all independent and so \cref{glss} recovers a well-known Chernoff bound for sums of independent Bernoulli random variables. \cref{glss} also recovers this Chernoff bound when $n = rm$ and $Y_j = X_{\ceil{j/r}}$ for all $j$~\cite{Gav+15}.
\end{rmk*}

Consider a string $x$ of independent Bernoulli random variables. If $G$ is a generalized Toffoli gate then $G\ket x$ is a read-2 family, because for all $i$ the $i$'th bit of $x$ can only influence the $i$'th and target bits of $G\ket x$. More generally, if $G$ is a generalized Toffoli gate and $L_1, L_2$ are layers of NOT gates acting on subsets of the support of $G$, then $L_1 G L_2 \ket x$ is a read-2 family. Even more generally, it follows by induction that if $C$ is a depth-$d$ purely classical circuit then $C\ket x$ is a read-$2^d$ family.

Before proving \cref{mc-nek-lb-nice}, as a warmup we briefly prove the following:

\begin{prp}\label{pc-prod}
	If $C$ is a depth-$d$ purely classical circuit and $\ket\iv$ is a mono-product state, then $|\bra\nu C\ket \iv|^2 \le 1/2 + \exp(-\Omega(n/2^d))$ for all $n$-nekomata $\ket\nu$.
\end{prp}

\begin{proof}
	Since standard-basis measurements of qubits in a mono-product state are independent, it follows from the above discussion that a standard-basis measurement of any $n$ designated target qubits of $C\ket\iv$ is a read-$2^d$ family. If the expected Hamming weight of a standard-basis measurement of the targets of $C\ket\iv$ is less (resp.\ greater) than or equal to $n/2$, then \cref{glss} implies that the targets of $C\ket\iv$ measure to all-ones (resp.\ all-zeros) with probability at most $\exp(-\Omega(n/2^d))$, and the result follows from \cref{nek-lb}.
\end{proof}

\begin{rst*}[\cref{mc-nek-lb-nice}]
	If $C$ is a nice, mostly classical circuit of size $s$ and depth $d$, then the fidelity of $C\ket\zs$ and any $n$-nekomata is at most
	\begin{equation*}
		\frac{1}{2} + \exp\left(-\Omega\left(\frac{n/2^d}{\max\left(\log s, \sqrt{n/2^d} \right)} \right) \right).
	\end{equation*}
\end{rst*}

\begin{proof}
	Designate $n$ qubits of $C\ket{\zs}$ as targets, and assume without loss of generality that $s \ge \exp\left(\sqrt{n/2^d}\right)$. We will prove that for some $a \in \bits$, the targets of $C\ket{\zs}$ measure to $a^n$ with probability at most $\exp(-\Omega(n2^{-d}/\log s))$. The result then follows from \cref{nek-lb}.
	
	Write $C = D(L \otimes \bigotimes_{G \in \mathcal G} G)$ such that $D$ is purely classical, $L$ is a layer of single-qubit gates, and $\mathcal G$ is a set of multi-qubit \rtt gates that each satisfy the precondition of \cref{convex-comb}. For all $G \in \mathcal G$, a standard-basis measurement of $G\ket{\zs}$ is distributed identically to $(b_G \wedge x_{G,i})_i$ for some independent Bernoulli random variables $b_G, (x_{G,i})_i$, where $\Ex[b_G] = 4\prod_i (1-\Ex[x_{G,i}])$. Let $\mu_G = \sum_i \Ex[x_{G,i}]$; then $\Ex[b_G] \le 4\exp(-\mu_G)$.
	
	By a union bound, the probability that there exists $G \in \mathcal G$ such that $\mu_G > 2 \ln s$ and $b_G = 1$ is at most
	\begin{equation*}
	\sum_{\mathclap{G : \mu_G > 2 \ln s}} 4\exp(-\mu_G) < 4s \exp(-2\ln s) = \exp(-\Omega(\log s)) \le \exp(-\Omega(n2^{-d}/\log s)).
	\end{equation*}
	Therefore it suffices to prove that for some $a \in \bits$, the targets of $\ket\aov \coloneqq D(L \otimes \bigotimes_{G: \mu_G \le 2\ln s} G \otimes I) \ket{\zs}$ measure to $a^n$ with probability at most $\exp(-\Omega(n2^{-d} /\log s))$. Henceforth we will never refer to any gate $G$ for which $\mu_G > 2 \ln s$; phrases such as ``for all $G$" and ``$(\cdot_G)_G$" will implicitly quantify over only those gates $G$ for which $\mu_G \le 2\ln s$.

	Let $b = (b_G)_G$ and $x = (x_{G,i})_{G,i}$. Call $x$ ``good" if $\sum_i x_{G,i} \le c\ln s$ for all $G$, where $c>2$ is a universal constant large enough so that $e(2e/c)^c < 1$. A well-known Chernoff bound states that if $S$ is a sum of independent Bernoulli random variables, and $\mu = \Ex[S]$, then $P(S > t) < (e\mu/t)^t e^{-\mu}$ for all $t>\mu$. Therefore, by a union bound and the fact that $\max_G \mu_G \le 2\ln s$, the probability that $x$ fails to be good is at most
	\begin{equation*}
	\sum_G (e\mu_G/c\ln s)^{c \ln s} \le s(2e/c)^{c \ln s} = (e(2e/c)^c)^{\ln s} = \exp(-\Omega(\log s)) \le \exp(-\Omega(n2^{-d}/\log s)).
	\end{equation*}
	
	Let $y$ be a string of independent Bernoulli random variables distributed identically to a standard-basis measurement of $L\ket{\zs}$. Call the targets of $D\ket{y, (b_G \wedge x_{G,i})_{G,i}, \zs}$ the ``output bits", and note that they are distributed identically to a standard-basis measurement of the targets of $\ket\aov$. If $b$ is fixed then the output bits are a read-$2^d$ family (as functions of the independent Bernoulli random variables in $x$ and $y$). Alternatively, if $x$ and $y$ are fixed and $x$ is good then the output bits are a read-$O(2^d \log s)$ family (as functions of the independent Bernoulli random variables in $b$).
	
	For $r_1, r_2 \in \R$ let $r_1 \approx r_2$ if $|r_1 - r_2| \le 0.1 n$. Let $W(b,z)$ be the Hamming weight of the output bits as a function of $b$ and $z \coloneqq (x,y)$, and let $z^\prime$ be an independent copy of $z$. We now argue that
	\begin{equation}\label{mceq1}
	P\left(W(b,z) \approx \Ex[W(b,\cdot) \mid b] \approx W(b,z^\prime) \approx \Ex[W(\cdot, z^\prime) \mid z^\prime]\right) \ge 1 - \exp(-\Omega(n2^{-d}/\log s)),
	\end{equation}
	where the probability is over $(b,z,z^\prime)$, and the expectations are over independent copies of $z$ and $b$ respectively, substituted for ``$\cdot$" as inputs to $W$. For any fixed value of $b$, \cref{glss} implies that $W(b,z) \approx \Ex[W(b,\cdot) \mid b]$ except with probability at most $\exp(-\Omega(n2^{-d}))$ over $z$, and the same statement holds with $z^\prime$ in place of $z$. Similarly, for any fixed value of $z^\prime = (x^\prime, y^\prime)$ such that $x^\prime$ is good, \cref{glss} implies that $W(b,z^\prime) \approx \Ex[W(\cdot, z^\prime) \mid z^\prime]$ except with probability at most $\exp(-\Omega(n2^{-d}/\log s))$ over $b$. Since $x^\prime$ is good except with probability at most $\exp(-\Omega(n2^{-d}/\log s))$, \cref{mceq1} follows from a union bound.
	
	Therefore, by the triangle inequality,
	\begin{equation*}
	P\left(\left|W(b,z) - \Ex[W(\cdot, z^\prime) \mid z^\prime]\right| \ge 0.3n \right) \le \exp(-\Omega(n2^{-d}/\log s)),
	\end{equation*}
	where the probability is over $(b,z,z^\prime)$. It follows that there exists a fixed value of $z^\prime$ such that the above inequality holds with the probability being over $(b,z)$. For this fixed value of $z^\prime$, if $\Ex[W(\cdot, z^\prime) \mid z^\prime]$ is at most (resp.\ at least) $n/2$, then the output bits are all-ones (resp.\ all-zeros) with probability at most $\exp(-\Omega(n2^{-d}/\log s))$.
\end{proof}

\section{Lower Bounds for General QAC Circuits}\label{more-lb}
In \cref{slbgc-sec} we prove a generalization of \cref{small-sc}. The proof uses the following claim, which is proved in \cref{it-proj-sec} (and is obtained as a corollary of a stronger result):
\begin{cor}\label{proj-tl}
	For all $d\ge1$, orthogonal projections $Q_1, \dotsc, Q_d$, and states $\ket\iv$,
	\begin{equation*}
		\norm{Q_d \dotsb Q_1 \ket\iv} \le \exp\left(-\frac{\bra\iv (I - Q_d) \ket\iv}{2d}\right).
	\end{equation*}
\end{cor}
Then, using this generalization of \cref{small-sc}, in \cref{meas-anc,lbcd2-sec} we prove \cref{d2-lb-cat}.
\subsection{Proof of \texorpdfstring{\cref{small-sc}}{Theorem 1.5}} \label{slbgc-sec}

\cref{small-sc} is the case of the following in which $\Hi_1, \dotsc, \Hi_n$ are single-qubit Hilbert spaces, $\ket\iv$ is the all-zeros state, $Q_j = \kb0$ for all $j$, and $\ket\dov$ is an $n$-nekomata.

\begin{thm}\label{small}
	There is a universal constant $c > 0$ such that the following holds. Let $\Hi_1, \dotsc, \Hi_n$ be Hilbert spaces, let $\Hi_T = \bigotimes_{j=1}^n \Hi_j$ (for ``targets"), and let $\Hi_A$ be a Hilbert space (for ``ancillae"). Let $\ket\iv = \ket{\iv_1, \dotsc, \iv_n, \iv_A}$ for some states $\ket{\iv_j} \in \Hi_j, j \in [n] \cup \{A\}$. Let $Q_j$ be an orthogonal projection on $\Hi_j$ for $j\in[n]$, and let $\ket\dov$ be a state in $\Hi_T \otimes \Hi_A$ that measures to $\bigotimes_{j=1}^n Q_j \otimes I_{\Hi_A}$ and to $\bigotimes_{j=1}^n (I-Q_j) \otimes I_{\Hi_A}$ each with probability 1/2. Let $C$ be a depth-$d$ QAC circuit on $\Hi_T \otimes \Hi_A$ with at most $cn/(d+1)$ multi-qubit gates acting on $\Hi_T$. Then, $|\bra \dov C \ket \iv|^2 \le 1/2 + \exp(-\Omega(n/(d+1)))$.
\end{thm}

\begin{proof}
	By \cref{rnf} we may write $C = DL$ for some layer of single-qubit gates $L$ and QAC circuit $D$, where $D$ has the same topology as $C$ and consists only of multi-qubit \rtt gates. Since $L\ket\iv$ factors as a product state in the same way that $\ket\iv$ does, we may assume without loss of generality that $C$ consists only of multi-qubit \rtt gates, by replacing $C$ and $\ket\iv$ with $D$ and $L\ket\iv$ respectively.
	
	We now generalize \cref{nek-lb} from nekomata to states such as $\ket\dov$. Let $Q = \bigotimes_{j=1}^n Q_j \otimes I_{\Hi_A}$ and $Q^\prime = \bigotimes_{j=1}^n (I-Q_j) \otimes I_{\Hi_A}$, and let $\ket\aov = C \ket\iv$. Since $\ket\dov$ measures to $Q + Q^\prime$ with probability 1, it follows from the triangle inequality and Cauchy-Schwarz that
	\begin{align*}
		|\ip \dov \aov|^2 &= |\bra \dov (Q + Q^\prime) \ket \aov|^2 \le (|\bra \dov Q \ket \aov| + |\bra \dov Q^\prime \ket \aov|)^2 \\
		&\le (\norm{Q \ket \aov} \cdot \norm{Q \ket \dov} + \norm{Q^\prime \ket \aov} \cdot \norm{Q^\prime \ket \dov})^2
		= (\norm{Q \ket \aov}/\sqrt 2 + \norm{Q^\prime \ket \aov}/\sqrt 2)^2 \\
		&= \bra \aov (Q + Q^\prime) \ket \aov/2 + \norm{Q \ket \aov} \cdot \norm{Q^\prime \ket \aov}
		\le 1/2 + \min(\norm{Q\ket\aov}, \norm{Q^\prime\ket\aov}),
	\end{align*}
	so it suffices to prove that $\min(\norm{Q\ket\aov}, \norm{Q^\prime\ket\aov}) \le \exp(-\Omega(n/(d+1)))$.
	
	Since $\sum_{j=1}^n \bra{\iv_j} Q_j \ket{\iv_j} + \sum_{j=1}^n \bra{\iv_j} (I - Q_j) \ket{\iv_j} = n$, either $\sum_{j=1}^n \bra{\iv_j} Q_j \ket{\iv_j} \ge n/2$ or $\sum_{j=1}^n \bra{\iv_j} (I-Q_j) \ket{\iv_j} \ge n/2$. Assume without loss of generality that $\sum_{j=1}^n \bra{\iv_j} (I-Q_j) \ket{\iv_j} \ge n/2$. We will prove that $\norm{Q \ket \aov} \le \exp(-\Omega(n/(d+1)))$.
	
	Let $\mathcal G$ be the set of gates in $C$, ordered such that $C = \prod_{G \in \mathcal G} G$ (where each gate $G$ is implicitly tensored with the identity). Also let $\mathcal G_T \subseteq \mathcal G$ be the set of gates in $C$ that act on $\Hi_T$. For $G \in \mathcal G_T$ let $\ket{\cv_G}$ be the mono-product state, specified up to a phase factor, such that $G = \rt{\cv_G} = I - 2\kb{\cv_G}$. Let $F$ be the set of functions with domain $\mathcal G$ that map each gate $G$ in $\mathcal G_T$ to either $I$ or $\kb{\cv_G}$, and map each gate $G$ in $\mathcal G \backslash\mathcal G_T$ to $G$ itself. Then $C = \sum_{f \in F} (-2)^{|\{G: f(G) = \kb{\cv_G}\}|} \prod_{G \in \mathcal G} f(G)$, so by the triangle inequality,
	\begin{equation*}
		\norm{Q \ket \aov} = 
		\norm{Q C \ket \iv} \le
		\sum_{f \in F} 2^{|\{G: f(G) = \kb{\cv_G}\}|} \cdot \max_{f \in F} \Norm{Q \prod_{\mathclap{G \in \mathcal G}} f(G) \cdot \ket \iv}.
	\end{equation*}
	
	By assumption, $|\mathcal G_T| \le cn/(d+1)$ (for a constant $c$ to be specified later), so
	\begin{equation*}
		\sum_{f \in F} 2^{|\{G: f(G) = \kb{\cv_G}\}|} = \sum_{\mathclap{S \subseteq \mathcal G_T}} 2^{|S|} = \prod_{\mathclap{G \in \mathcal G_T}} (2^0 + 2^1) = 3^{|\mathcal G_T|} \le 3^{cn/(d+1)}.
	\end{equation*}
	Consider an arbitrary function $f \in F$. For all $G \in \mathcal G$ we may write $f(G) = f_T(G) \otimes f_A(G)$, where $f_T(G)$ is a tensor product of one-qubit orthogonal projections on $\Hi_T$, and $f_A(G)$ is either an orthogonal projection or a unitary transformation on $\Hi_A$. (This can be seen by considering all three cases: $f(G) = I$, $f(G) = \kb{\cv_G}$, or $G \notin \mathcal G_T$ and $f(G) = G$.) Furthermore, if $G \notin \mathcal G_T$ then $f_T(G) = I$. Therefore, letting $\ket{\iv_T} = \ket{\iv_1, \dotsc, \iv_n}$,
	\begin{equation*}
		\Norm{Q \prod_{G \in \mathcal G} f(G) \cdot \ket \iv}
		= \Norm{\bigotimes_j Q_j \cdot \prod_{\mathclap{G \in \mathcal G_T}} f_T(G) \cdot \ket{\iv_T}} \cdot \Norm{\prod_{G \in \mathcal G} f_A(G) \cdot \ket{\iv_A}}.
	\end{equation*}
	Clearly $\Norm{\prod_{G \in \mathcal G} f_A(G) \cdot \ket{\iv_A}} \le 1$. For $k \in [d]$ let $M_k$ be the tensor product of $f_T(G)$ over all ``depth-$k$" gates $G \in \mathcal G_T$, i.e.\ $M_1, \dotsc, M_d$ are layers of one-qubit orthogonal projections such that $\prod_{G \in \mathcal G_T} f_T(G) = M_d \dotsb M_1$. Write $M_k = \bigotimes_{j=1}^n M_{jk}$, where $M_{jk}$ is an orthogonal projection on $\Hi_j$ for all $j \in [n]$. Then, by \cref{proj-tl},
	\begin{align*}
		\Norm{\bigotimes_j Q_j \cdot \prod_{\mathclap{G \in \mathcal G_T}} f_T(G) \cdot \ket{\iv_T}}
		&= \prod_{j=1}^n \Norm{Q_j M_{jd} \dotsb M_{j1}\ket{\iv_j}}
		\le \prod_{j=1}^n \exp\left(-\frac{\bra{\iv_j}(I - Q_j)\ket{\iv_j}}{2(d+1)}\right) \\
		&\le \exp\left(-\frac{n/2}{2(d+1)}\right).
	\end{align*}
	Altogether this implies that $\norm{Q \ket \aov} \le \exp((c\ln 3 - 1/4) \cdot n/(d+1))$, and the result follows by taking $c < 1/(4\ln 3)$.
\end{proof}
\subsection{Proof of \texorpdfstring{\cref{proj-tl}}{Corollary 5.1}} \label{it-proj-sec}

Let $\Delta(\ket\alpha,\ket\beta) = \arccos|\ip\alpha\beta|$; we will abbreviate this as $\Delta(\alpha,\beta)$.
\begin{lem}\label{dtri}
	The function $\Delta$ satisfies the triangle inequality, i.e.\ $\Delta(\alpha,\gamma) \le \Delta(\alpha,\beta) + \Delta(\beta,\gamma)$ for all states $\ket\alpha, \ket\beta, \ket\gamma$.
\end{lem}
\begin{rmk*}
	The intuition behind our ultimate use of \cref{dtri} is that, up to normalization, the total amount of ``progress" made by $Q_{d-1} \dotsb Q_1$ in interpolating between $\ket\iv$ and $Q_d$ is at most the sum of the amounts of progress made by the individual projections $Q_1, \dotsc, Q_{d-1}$.

	For intuition as to why \cref{dtri} is true, consider the similarly defined function $\Delta^\prime(u,v) = \arccos\langle u,v\rangle$ for unit vectors $u,v \in \R^3$, where $\langle\cdot,\cdot\rangle$ denotes the usual inner product on $\R^3$. It is well known that $\Delta^\prime(u,v)$ equals the angle between $u$ and $v$, which equals the length of the arc (\cref{sphere}) formed by traversing a great circle on the unit sphere from $u$ to $v$ in the shorter of the two directions. This arc is known to be the shortest path on the unit sphere between $u$ and $v$, so $\Delta^\prime$ represents distance on the unit sphere.
	
	We make two more unrelated remarks. First, if we count states differing only by a phase factor as equivalent, then \cref{dtri} implies that $\Delta$ is a metric on the set of states (on some consistent number of qubits). Second, the results in this subsection generalize easily from $\C^{2^n}$ to arbitrary Hilbert spaces.
\end{rmk*}
\begin{figure}
	\begin{minipage}[b]{0.3\textwidth}
		\centering
		\begin{tikzpicture}
			[vec/.style = {semithick, arrows = {-Stealth[scale=1.5, angle'=30]}}]
			\pgfmathsetmacro{\r}{2}
			\pgfmathsetmacro{\h}{0.8}
			\pgfmathsetmacro{\u}{220}
			\pgfmathsetmacro{\v}{320}
			\pgfmathsetmacro{\p}{0.85}
			\pgfmathsetmacro{\e}{2}
			\node[above] (ul) at ({\p*\r*cos(\u)}, {\p*\h*sin(\u)}) {$u$};
			\node[above] (vl) at ({\p*\r*cos(\v)}, {\p*\h*sin(\v)}) {$v$};
			\draw[very thin] (0,0) circle (\r);
			\shade[ball color = white, very thin, opacity = 0.25] (0,0) circle (\r);
			\draw[thin,dashed,domain=0:180] plot ({\r*cos(\x)}, {\h*sin(\x)});
			\draw[thin,domain=180:{\u-\e}] plot ({\r*cos(\x)}, {\h*sin(\x)});
			\draw[red, thick, domain=\u:\v] plot ({\r*cos(\x)}, {\h*sin(\x)});
			\draw[thin,domain={\v+\e}:360] plot ({\r*cos(\x)}, {\h*sin(\x)});
			\draw[vec] (0,0) to ({\r*cos(\u)}, {\h*sin(\u)}) ;
			\draw[vec] (0,0) to ({\r*cos(\v)}, {\h*sin(\v)});
			\draw[fill=black] (0,0) circle (0.05);
		\end{tikzpicture}
		\caption{}
		\label{sphere}
	\end{minipage}
	\hfill
	\begin{minipage}[b]{0.6\textwidth}
		\begin{tikzpicture}
			\foreach \j in {0,1,...,4}
			\draw[->] (0,0) to node[near end, above  left = -5.2pt]{$\ket{\cv_\j}$} (\j/4*70:4);
		\end{tikzpicture}
		\caption{An optimal choice of $\ket{\cv_1}, \dotsc, \ket{\cv_{d-1}}$ in the $d=4$ case of \cref{proj1}.}
		\label{proj-fig}
	\end{minipage}
\end{figure}
\begin{proof}
	Let $\ket\dov$ be a state orthogonal to $\ket\alpha$ such that $\ket\gamma$ is in the span of $\ket\alpha$ and $\ket\dov$, and let $\eta = \arccos \frac{|\ip\beta\alpha|}{\sqrt{|\ip\beta\alpha|^2 + |\ip\beta\dov|^2}}$. By the triangle inequality,
	\begin{align*}
		|\ip\beta\gamma| &= |\bra\beta(\kb\alpha+\kb\dov)\ket\gamma|
		\le |\ip\beta\alpha| \cdot |\ip\gamma\alpha| + |\ip\beta\dov| \cdot |\ip\gamma\dov| \\
		&\le \frac{|\ip\beta\alpha|}{\sqrt{|\ip\beta\alpha|^2 + |\ip\beta\dov|^2}} \cdot |\ip\gamma\alpha| +
		\frac{|\ip\beta\dov|}{\sqrt{|\ip\beta\alpha|^2 + |\ip\beta\dov|^2}} \cdot |\ip\gamma\dov| \\
		&= \cos\eta \cdot \cos \Delta(\alpha,\gamma) + \sin\eta \cdot \sin \Delta(\alpha,\gamma)
		= \cos(\Delta(\alpha,\gamma)-\eta),
	\end{align*}
	so $\Delta(\beta,\gamma) \ge |\Delta(\alpha,\gamma) - \eta| \ge \Delta(\alpha,\gamma) - \eta$. Similarly,
	$|\ip\beta\alpha| \le
	\frac{|\ip\beta\alpha|}{\sqrt{|\ip\beta\alpha|^2 + |\ip\beta\dov|^2}}
	= \cos\eta$,
	so $\Delta(\alpha,\beta) \ge \eta$. Therefore $\Delta(\alpha,\beta) + \Delta(\beta,\gamma) \ge \eta + (\Delta(\alpha,\gamma) - \eta) = \Delta(\alpha,\gamma)$.
\end{proof}
\begin{prp}\label{strong-tl}
	For all $d\ge1$, nonzero orthogonal projections $Q_d$, and states $\ket\iv$,
	\begin{equation*}
		\max_{\mathclap{Q_1, \dotsc, Q_{d-1}}} \norm{Q_d Q_{d-1} \dotsb Q_1 \ket\iv} 
		= \cos\left(\frac{\arccos \norm{Q_d \ket\iv}} d\right)^d,
	\end{equation*}
	where the maximum is taken over all orthogonal projections $Q_1, \dotsc, Q_{d-1}$.
\end{prp}
\begin{proof}
	We first prove an analogous statement about rank-1 orthogonal projections, specifically that for all states $\ket{\cv_0}$ and $\ket{\cv_d}$,
	\begin{equation}\label{proj1}
		\max_{\ket{\cv_1}, \dotsc, \ket{\cv_{d-1}}} \left|\prod_{j=1}^d \ip{\cv_{j-1}}{\cv_j}\right|
		= \cos\left(\frac{\arccos|\ip{\cv_0}{\cv_d}|}d\right)^d.
	\end{equation}
	We then prove that the original proposition follows from this rank-1 analogue.
	
	First we prove that the left side of \cref{proj1} is at most the right side. On the image of $\Delta$, i.e.\ on the interval $[0,\pi/2]$, the cosine function is decreasing and concave. Therefore for all states $\ket{\cv_1}, \dotsc, \ket{\cv_{d-1}}$, by the AM-GM inequality, Jensen's inequality, and \cref{dtri},
	\begin{align*}
		\left|\prod_{j=1}^d \ip{\cv_{j-1}}{\cv_j}\right|^{1/d}
		&\le \frac{1}d \sum_{j=1}^d |\ip{\cv_{j-1}}{\cv_j}|
		= \frac{1}d \sum_{j=1}^d \cos \Delta(\cv_{j-1},\cv_j)
		\le \cos\left(\frac{1}d\sum_{j=1}^d \Delta(\cv_{j-1},\cv_j)\right) \\
		&\le \cos\left(\frac{\Delta(\cv_0,\cv_d)}d\right)
		= \cos\left(\frac{\arccos|\ip{\cv_0}{\cv_d}|}d\right).
	\end{align*}

	Next we give an example (\cref{proj-fig}) which shows that the left side of \cref{proj1} is at least the right side. (This part is not needed to prove \cref{proj-tl}, but it is brief and may be of independent interest.) For ease of notation let $\ket\sigma = \ket{\cv_0}$ and $\ket\tau = \ket{\cv_d}$. By multiplying $\ket\tau$ by a phase factor we may assume that $\ip\sigma\tau$ is a nonnegative real number. Let $\eta = \arccos(\ip\sigma\tau)/d$, let
	$\ket\dov =
	\frac{(I - \kb\sigma)\ket\tau}{\norm{(I - \kb\sigma)\ket\tau}} =
	\frac{\ket \tau  - \ket \sigma \ip \sigma \tau }{\sqrt{1 - \ip \sigma \tau ^2}}$,
	and for $j\in[d-1]$ let $\ket{\cv_j} = \cos(j\eta)\ket\sigma + \sin(j\eta)\ket\dov$. The latter equation also holds for $j=0$ and $j=d$, respectively because $\ket\sigma = \ket{\cv_0}$ and
	\begin{equation*}
		\cos(d\eta) \ket\sigma + \sin(d\eta) \ket\dov = \ip\sigma\tau \cdot \ket\sigma + \sqrt{1-\ip\sigma\tau^2} \cdot \ket\dov = \ket\tau = \ket{\cv_d}.
	\end{equation*}
	Since $\ip\sigma\dov=0$, it follows that for all $j\in[d]$,
	\begin{equation*}
		\ip{\cv_{j-1}}{\cv_j} = \cos((j-1)\eta)\cos(j\eta) + \sin((j-1)\eta)\sin(j\eta) = \cos(j\eta - (j-1)\eta) = \cos(\eta),
	\end{equation*}
	so $\prod_{j=1}^d\ip{\cv_{j-1}}{\cv_j} = \cos(\eta)^d$ as desired.
	
	Finally we prove that the original proposition follows from \cref{proj1}. For $j \in [d-1]$ we may assume that $Q_j$ is rank-1, because  if $Q_j \dotsb Q_1 \ket\iv \neq 0$ then $Q_j \dotsb Q_1 \ket\iv = \kb{\cv_j} Q_{j-1} \dotsb Q_1 \ket\iv$ for $\ket{\cv_j} = \frac{Q_j \dotsb Q_1 \ket\iv}{\norm{Q_j \dotsb Q_1 \ket\iv}}$, and if $Q_j \dotsb Q_1 \ket\iv = 0$ then clearly we cannot decrease $\norm{Q_d \dotsb Q_1 \ket\iv}$ by replacing $Q_j$ with an arbitrary rank-1 orthogonal projection. For any state $\ket \aov$, the norm $\norm{Q_d \ket\aov}$ equals the maximum of $|\ip\dov \aov|$ over all states $\ket\dov$ such that $Q_d\ket\dov = \ket\dov$.\footnote{By Cauchy-Schwarz, $|\ip\dov \aov| = |\bra\dov Q_d \ket\aov| \le \norm{Q_d \ket\aov}$, with equality if $\ket\dov = Q_d \ket\aov / \norm{Q_d \ket\aov}$ or if $Q_d\ket\aov = 0$.} (Here we used the fact that $Q_d \neq 0$ to ensure that there exists a \emph{state} in the 1-eigenspace of $Q_d$, rather than just the zero vector.) Therefore,
	\begin{align*}
		\max_{\mathclap{Q_1, \dotsc, Q_{d-1}}} \norm{Q_d \dotsb Q_1 \ket\iv}
		&= \;\max_{\mathclap{\ket{\cv_1}, \dotsc, \ket{\cv_{d-1}}}} \norm{Q_d \ket{\cv_{d-1}} \dotsb \ip{\cv_1}\iv}
		= \;\max_{\mathclap{\substack{\ket{\cv_1}, \dotsc, \ket{\cv_{d-1}} \\ \ket\dov = Q_d\ket\dov}}} \left|\ip\dov{\cv_{d-1}} \dotsb \ip{\cv_1}\iv\right| \\
		&= \;\max_{\mathclap{\ket\dov = Q_d \ket\dov}} \cos\left(\frac{\arccos|\ip\dov\iv|} {d}\right)^d
		= \cos\left(\frac{\arccos \left(\max_{\ket\dov = Q_d \ket\dov} |\ip\dov\iv|\right)}{d} \right)^d \\
		&= \cos\left(\frac{\arccos \norm{Q_d \ket\iv}} d\right)^d. \qedhere
	\end{align*}
\end{proof}
\begin{rst*}[\cref{proj-tl}]
	For all $d \ge 1$, orthogonal projections $Q_1, \dotsc, Q_d$, and states $\ket\iv$,
	\begin{equation*}
		\norm{Q_d \dotsb Q_1 \ket\iv} \le \exp\left(-\frac{\bra\iv (I - Q_d) \ket\iv}{2d}\right).
	\end{equation*}
\end{rst*}
\begin{proof}
	The claim is trivial if $Q_d=0$, so assume otherwise. Since
	\begin{equation*}
		\arccos \norm{Q_d \ket\iv} \ge \sin \arccos \norm{Q_d \ket\iv} = \sqrt{1 - \norm{Q_d \ket\iv}^2} = \sqrt{\bra\iv (I - Q_d) \ket\iv},
	\end{equation*}
	it follows from \cref{strong-tl} that
	\begin{equation*}
		\norm{Q_d \dotsb Q_1 \ket\iv} \le
		\cos\left(\frac{\arccos \norm{Q_d \ket\iv}} d\right)^d \le
		\cos\left(\frac{\sqrt{\bra\iv (I - Q_d) \ket\iv}} d\right)^d,
	\end{equation*}
	so it suffices to prove that $\cos r \le \exp(-r^2/2)$ for all $r \in [0,1]$.
	
	A special case of the Lagrange remainder theorem states that if $f:\R\to\R$ is $n$ times differentiable on all of $\R$, then for all $x \in \R$ there exists $h$ between 0 and $x$ such that
	\begin{equation*}
		f(x) = \sum_{k=0}^{n-1} \frac{f^{(k)}(0)}{k!} x^k + \frac{f^{(n)}(h)}{n!} x^n,
	\end{equation*}
	where $f^{(k)}$ denotes the $k$'th derivative of $f$. An application with $f = \cos(\cdot), x = r, n = 4$ reveals that
	\begin{equation*}
		\cos r \le
		1 - \frac{r^2}{2} + \frac{\max\{\cos h: 0 \le h \le r\}}{24} \cdot r^4 =
		1 - \frac{r^2}{2} + \frac{r^4}{24},
	\end{equation*}
	and an application with $f = \exp(\cdot), x = -r^2/2, n = 3$ reveals that
	\begin{equation*}
		e^{-r^2/2} \ge 1 - r^2/2 + \frac{1}{2} (-r^2/2)^2 + \frac{\max\{e^h: -r^2/2 \le h \le 0\}}{6} (-r^2/2)^3 \\
		= 1 - \frac{r^2}{2} + \frac{r^4}{8} - \frac{r^6}{48}.
	\end{equation*}
	Finally, since $r^2 \le 1$ it follows that $r^6 \le r^4$, so
	\begin{equation*}
		e^{-r^2/2} \ge 1 - \frac{r^2}{2} + \frac{r^4}{8} - \frac{r^4}{48} \ge 1 - \frac{r^2}{2} + \frac{r^4}{24} \ge \cos r. \qedhere
	\end{equation*}
\end{proof}

\subsection{Simplifying Depth-2 QAC Circuits by Measuring Ancillae}\label{meas-anc}

For a one-qubit state $\ket\dov$, let \emph{the $\ket\dov$ basis} be an orthonormal basis of $\C^2$ that includes $\ket\dov$. (We refer to ``the" $\ket\dov$ basis because, up to a phase factor, there is a unique state orthogonal to $\ket\dov$.)

\begin{lem}\label{meas-z}
	Let $\Hi_1$ be a one-qubit Hilbert space, and let $\Hi_2$ and $\Hi_3$ be Hilbert spaces on arbitrary numbers of qubits. Then for all $\ket\dov \in \Hi_1, \ket\cv \in \Hi_2, \ket\iv \in \Hi_1 \otimes \Hi_2 \otimes \Hi_3$, the following two procedures generate identically distributed random states in $\Hi_1 \otimes \Hi_2 \otimes \Hi_3$:
	\begin{itemize}[leftmargin=*]
		\renewcommand\labelitemi{--}
		\item measure the $\Hi_1$ qubit of $(\rt{\ket{\dov,\cv}} \otimes I_{\Hi_3}) \ket\iv$ in the $\ket\dov$ basis;
		\item measure the $\Hi_1$ qubit of $\ket\iv$ in the $\ket\dov$ basis, and then, conditioned on the outcome being $\ket\dov$, apply $\rt\cv$ on $\Hi_2$.
	\end{itemize}
\end{lem}
\begin{proof}
	This follows easily from the fact that $\rt{\ket{\dov,\cv}} = (I - \kb\dov) \otimes I + \kb\dov \otimes \rt\cv$.
\end{proof}

\cref{d2-lb-cat} is clearly equivalent to the statement that if $C$ is a depth-2 QAC circuit, then any $n$ designated ``target" qubits of $C\ket{\zs}$ measure to $\ket{\Cat_n}$ with probability at most $1/2 + \exp(-\Omega(n))$. The following is the starting point for our proof:

\begin{prp}\label{d2-small-wlog}
	Let $p$ and $\ket\dov$ be such that for some depth-2 QAC circuit $C$, designated ``target" qubits of $C\ket{\zs}$ measure to $\ket\dov$ with probability $p$. Then there exist layers of \rtt gates $L_2, L_1$ and a mono-product state $\ket\iv$ such that for some partition of the qubits of $L_2 L_1 \ket\iv$ into ``targets" and ``ancillae",
	\begin{enumerate}[label=(\roman*)]
		\item the targets of $L_2 L_1 \ket\iv$ measure to $\ket\dov$ with probability at least $p$;
		\item for all $k \in \{1,2\}$, every ancilla is acted on by a gate in $L_k$, and every gate in $L_k$ acts on at least one target.
	\end{enumerate}
\end{prp}

\begin{rmk*}
	Although not necessary for our purposes, using \cref{rnf} it is easy to generalize the following argument to show that the gates in $L_2$ and $L_1$ may be assumed to be multi-qubit gates.
\end{rmk*}

\begin{proof}
	Let a ``construction" be a tuple of the form $(L_2, L_1, \ket\iv)$ where $L_2$ and $L_1$ are layers of \rtt gates and $\ket\iv$ is a mono-product state. By \cref{rnf} there exists a construction satisfying (i). Below we describe a procedure that takes as input a construction satisfying (i) but not (ii), and outputs a construction satisfying (i) using fewer ancillae than the original construction. It then suffices to iterate this procedure on a construction satisfying (i) until the construction also satisfies (ii), because the number of ancillae can only decrease finitely many times.
	
	Let $(L_2, L_1, \ket\iv)$ be a construction satisfying (i) but not (ii), and let $\ket\aov = L_2 L_1 \ket\iv$. For all $k \in \{1,2\}$ and gates $G$ in $L_k$, write $G = \rt{\bigotimes_{\Hi} \ket{\cv_{\Hi}^k}}$, where $\Hi$ ranges over all one-qubit Hilbert spaces acted on by $G$, and $\ket{\cv_{\Hi}^k}$ is a state in $\Hi$. (Since $\Hi$ and $k$ uniquely determine $G$, this does not assign conflicting definitions to any of the $\ket{\cv_{\Hi}^k}$.)
	
	First consider the case where an ancilla $\Hi$ is not acted on by $L_2$ (that is, by any gate in $L_2$). If $\Hi$ is also not acted on by $L_1$ then we may simply remove $\Hi$ from the construction. Otherwise, measure the $\Hi$ qubit of $\ket\aov$ in the $\ket{\cv_\Hi^1}$ basis. By \cref{meas-z}, the resulting state on the qubits besides $\Hi$ equals $L_2^\prime L_1^\prime \ket{\iv^\prime}$ for some random construction $(L_2^\prime, L_1^\prime, \ket{\iv^\prime})$. Furthermore, the expectation over $(L_2^\prime, L_1^\prime, \ket{\iv^\prime})$ of the probability that the targets of $L_2^\prime L_1^\prime \ket{\iv^\prime}$ measure to $\ket\dov$ equals the probability that the targets of $\ket\aov$ measure to $\ket\dov$, which is at least $p$. Therefore there exists a fixed construction in the support of $(L_2^\prime, L_1^\prime, \ket{\iv^\prime})$ that satisfies (i), and the procedure may output this construction.
	
	If an ancilla $\Hi$ is acted on by $L_2$ but not by $L_1$, then measure the $\Hi$ qubit of $\ket\aov$ in the $\ket{\cv_\Hi^2}$ basis, and the rest of the argument is similar to the above. If every ancilla is acted on by $L_2$, and a gate $G$ in $L_1$ does not act on any targets, then for all qubits $\Hi$ acted on by $G$, measure the $\Hi$ qubit of $\ket\aov$ in the $\ket{\cv_\Hi^2}$ basis, and again the rest of the argument is similar to the above. Finally, if a gate $G$ in $L_2$ does not act on any targets, then $G$ acts on at least one ancilla, and also we may remove $G$ from $L_2$ without changing the probability that the targets of $L_2 L_1 \ket\iv$ measure to $\ket\dov$, so this reduces to the previously considered case in which an ancilla is not acted on by $L_2$.
\end{proof}

\subsection{Proof of \texorpdfstring{\cref{d2-lb-cat}}{Theorem 1.7(iii)}} \label{lbcd2-sec}

The $\delta=1$ case of the following is Markov's inequality:
\begin{lem}\label{gen-markov}
	Let $0 < \delta \le 1$, let $a > 0$, and let $X$ be a nonnegative random variable. Then there exists $t \in [a, ae^{\delta^{-1}-1}]$ such that $P(X \ge t) \le \delta \Ex[X]/t$.
\end{lem}

\begin{rmk*}
	The intuition behind our use of \cref{gen-markov} is as follows. \cref{small} implies that depth-2 QAC circuits require size at least $\Omega(n)$ to approximately construct $\ket{\Cat_n,\dov}$, and \cref{d2-small-wlog} implies that depth-2 QAC circuits that approximately construct $\ket{\Cat_n,\dov}$ have size at most $2n$ without loss of generality, so these bounds are ``just a constant factor" away from implying that depth-2 QAC circuits of arbitrary size cannot approximately construct $\ket{\Cat_n,\dov}$. This is analogous to how Markov's inequality is ``just a factor of $\delta$" away from the conclusion of \cref{gen-markov}.
\end{rmk*}

\begin{proof}
	Assume the contrary, and let $b = ae^{\delta^{-1}- 1}$. Then,
	\begin{equation*}
	\Ex[X]
	= \int_0^\infty P(X \ge t) dt
	\ge \int_0^a P(X \ge t) dt + \int_a^b P(X \ge t) dt,
	\end{equation*}
	and
	\begin{equation*}
	\int_0^a P(X \ge t) dt \ge \int_0^a P(X \ge a) dt = aP(X \ge a) > \delta \Ex[X],
	\end{equation*}
	and
	\begin{equation*}
	\int_a^b P(X \ge t) dt > \int_a^b \delta \Ex[X]/t \cdot dt = \delta \Ex[X] \ln(b/a) = \delta \Ex[X](\delta^{-1} - 1),
	\end{equation*}
	so $\Ex[X] > \Ex[X]$, which is a contradiction.
\end{proof}

\begin{thm}[Tur\'an's theorem\footnote{Usually Tur\'an's theorem is phrased as saying that dense graphs have large cliques, whereas \cref{turan} says that sparse graphs have large independent sets. These statements are equivalent, because taking the complement of a graph turns cliques into independent sets and vice versa.}]\label{turan}
	Let $\mathcal G$ be a simple undirected graph on $n$ vertices, and let $d$ be the average degree of the vertices in $\mathcal G$. Then $\mathcal G$ contains an independent set of size at least $n/(d+1)$.
\end{thm}

\begin{rmk*}
	For the intuition behind our use of \cref{turan}, recall the discussion of disjoint light cones from \cref{d2lb-intro}.
\end{rmk*}

\begin{proof}[Proof exposited by Alon and Spencer~\cite{AS16}]
	Identify the vertex set of $\mathcal G$ with $[n]$. Let $\sigma$ be a uniform random permutation of $[n]$, and let $\mathcal I$ be the set of vertices $u$ such that $\sigma(u) < \sigma(v)$ for all edges $\{u,v\}$. Then $\mathcal I$ is an independent set, because for every edge $\{u,v\}$, either $\sigma(u) < \sigma(v)$ or $\sigma(v) < \sigma(u)$. A vertex $u$ with degree $d_u$ is in $\mathcal I$ with probability $1/(d_u + 1)$, because any vertex out of $u$ and its neighbors is equally likely to be assigned the lowest value by $\sigma$ out of these vertices. By linearity of expectation it follows that $\Ex |\mathcal I| = \sum_{u \in [n]} 1/(d_u + 1)$, and by Jensen's inequality this is at least $n/(d+1)$.
\end{proof}

Recall that $\ket\iv, C, \ket\dov, (Q_j)_j$ are variables from the statement of \cref{small}. In upcoming applications of \cref{small} we will refer to $\ket\iv$ as the ``input state", $C$ as the ``circuit", $\ket\dov$ as the ``desired output state", and $(Q_j)_j$ as ``projections".

\begin{rmk*}
	We will not actually use the full strength of \cref{small}, in the sense that we will always upper-bound the number of multi-qubit gates acting on the targets by upper-bounding the \emph{total} number of gates. One could instead use the full strength of \cref{small} in this regard, and forgo the use of \cref{d2-small-wlog} entirely by measuring selected ancillae all at once later in the proof, but we consider the current presentation to be simpler.
\end{rmk*}

\begin{rst*}[\cref{d2-lb-cat}, paraphrased]
	If $C$ is a depth-2 QAC circuit, then any $n$ designated ``target" qubits of $C\ket{\zs}$ measure to $\ket{\Cat_n}$ with probability at most $1/2 + \exp(-\Omega(n))$.
\end{rst*}

\begin{proof}
	Let $L_2,L_1$ be layers of \rtt gates and let $\ket\iv$ be a mono-product state, with $n$ qubits designated as targets and all other qubits designated as ancillae. Assume that for all $k \in \{1,2\}$, every ancilla is acted on by a gate in $L_k$, and every gate in $L_k$ acts on at least one target. By \cref{d2-small-wlog} it suffices to prove that the targets of $L_2 L_1 \ket\iv$ measure to $\ket{\Cat_n}$ with probability at most $1/2 + \exp(-\Omega(n))$.
	
	Let $c$ be the constant from \cref{small}, and let $\gamma = (c/2)(c/3)/(1+c/2)$ and $\delta = (c/2)\gamma^2$. Since \cref{small} remains true if $c$ is replaced by any constant between 0 and $c$, we may take $c$ to be small enough so that $\gamma, \delta \le 1$.

	For a circuit $C$ let $|C|$ denote the number of gates in $C$, and write ``$G \in C$" to denote that $G$ is a gate in $C$. First consider the case where $|L_2| \le \gamma n$. It suffices to prove that $L_2 L_1 \ket\iv$ and $\ket{\Cat_n, \dov}$ have fidelity at most $1/2 + \exp(-\Omega(n))$ for all states $\ket\dov$. If $|L_1| \le n (c/3) / (1+c/2)$ then $|L_1| + |L_2| \le (c/3)n$, and the result follows from applying \cref{small} with input state $\ket\iv$, circuit $L_2 L_1$, desired output state $\ket{\Cat_n,\dov}$, and $n$ one-qubit projections $\kb0$ acting on the targets. Alternatively, if $|L_1| \ge n (c/3) / (1+c/2)$ then $|L_2| \le (c/2)|L_1|$, and the result follows from applying \cref{small} with input state $L_1 \ket\iv$, circuit $L_2$, desired output state $\ket{\Cat_n,\dov}$, and for every gate $G \in L_1$ the projection $\kb0 \otimes I$ on the support of $G$, where $\kb0$ acts on one of the targets acted on by $G$. (Here we used the fact that $1/2 + \exp(-\Omega(|L_1|)) \le 1/2 + \exp(-\Omega(n))$ by our assumption about $|L_1|$.)
	
	Now consider the case where $|L_2| \ge \gamma n$. This time we will measure some carefully chosen ancillae before applying \cref{small}. Let $X$ be the number of targets acted on by a uniform random gate in $L_1$. By \cref{gen-markov} there exists a real number $t \in [1,\exp(1/\delta)]$ such that $P(X \ge t) \le \delta \Ex[X]/t$. Fix such a $t$. Write $L_1 = L_1^B \otimes L_1^S$, for ``big" and ``small" respectively, where $L_1^B$ (resp.\ $L_1^S$) consists of the gates in $L_1$ acting on at least (resp.\ fewer than) $t$ targets. Then, $|L_1^B| = |L_1| P(X \ge t) \le \delta |L_1| \Ex[X]/t \le \delta n/t = (c/2)\gamma^2 n/t$.
	
	Let $\mathcal G$ be the undirected graph whose vertices are the gates in $L_2$, and whose edges are the pairs $e$ of distinct vertices such that for some gate $G \in L_1^S$, for both vertices $V$ in $e$, there exists a target that both $G$ and $V$ act on. Since $t \ge 1$, the degree of a vertex is at most $t-1$ times the number of targets acted on by that vertex. Therefore the average degree of the vertices in $\mathcal G$ is at most $(t-1)n/|L_2|$, so by \cref{turan} there exists an independent set $\mathcal I$ in $\mathcal G$ of size
	\begin{equation*}
	|\mathcal I| \ge \frac{|L_2|}{(t-1)n/|L_2| + 1}
	\ge \frac{\gamma n}{(t-1)n/(\gamma n) + 1}
	= \frac{\gamma^2 n}{t-1+\gamma}
	\ge \gamma^2 n / t.
	\end{equation*}
	Fix such a set $\mathcal I$. It follows that $|L_1^B| \le (c/2)|\mathcal I|$, and also that $|\mathcal I| \ge \gamma^2 n/ \exp(1/\delta) \ge \Omega(n)$.
	
	For $V \in \mathcal I$ let $\Hi_V$ be the Hilbert space consisting of the following two types of qubits: targets acted on by a gate in $L_1^S$ that acts on one of the same targets as $V$, and qubits acted on by $V$ that are not acted on by $L_1^S$. The $\Hi_V$ are Hilbert spaces on pairwise disjoint sets of qubits, because $\mathcal I$ is an independent set in $\mathcal G$ and because a qubit cannot be acted on by multiple gates in any given layer.
	
	For $G \in L_2$ write $G = \rt{\bigotimes_\Hi \ket{\cv_\Hi}}$, where $\Hi$ ranges over all one-qubit Hilbert spaces acted on by $G$, and $\ket{\cv_\Hi}$ is a state in $\Hi$. This defines $\ket{\cv_\Hi}$ for every ancilla $\Hi$, because $L_2$ acts on every ancilla. For all ancillae $\Hi$ acted on by $L_1^S$, measure the $\Hi$ qubit of $L_2 L_1 \ket\iv$ in the $\ket{\cv_\Hi}$ basis. By \cref{meas-z}, the resulting state $\ket\aov$ on the qubits that were not measured satisfies $\ket\aov = L_2^\prime L_1^B \ket{\iv^\prime}$, where $L_2^\prime$ and $L_1^B$ are implicitly tensored with the identity, and
	\begin{itemize}[leftmargin=*]
		\renewcommand\labelitemi{--}
		\item $\ket{\iv^\prime}$ is the tensor product of (i) a mono-product state on the qubits that were not acted on by $L_1^S$, and (ii) the tensor product over $G \in L_1^S$ of a state on the targets that were acted on by $G$. In particular, $\ket{\iv^\prime}$ factors as $\ket{\iv^\prime} = \bigotimes_{V \in \mathcal I} \ket{\iv^\prime_V} \otimes \ket{\iv^\prime_A}$, for some states $\ket{\iv_V^\prime} \in \Hi_V$ (none of the qubits in $\Hi_V$ were measured) and a state $\ket{\iv^\prime_A}$ on all other qubits in $\ket\aov$.
		\item $L_2^\prime = \bigotimes_{G \in L_2} U_G$, where $U_G$ is a Hermitian unitary transformation (specifically, the identity or an \rtt gate) on the qubits in $\ket\aov$ that were acted on by $G$.
	\end{itemize}

	It suffices to prove that the targets of $\ket\aov$ measure to $\ket{\Cat_n}$ with probability at most $1/2 + \exp(-\Omega(n))$, or equivalently that $|\bra{\Cat_n, \dov} L_2^\prime L_1^B \ket{\iv^\prime}|^2 \le 1/2 + \exp(-\Omega(n))$ for all states $\ket\dov$. For $V \in \mathcal I$, the transformation $U_V$ acts on a subset of the qubits in $\Hi_V$, including at least one target because $V$ acted on at least one target and none of the targets were measured. Therefore we may define an orthogonal projection on $\Hi_V$ by $Q_V = U_V(\kb0 \otimes I)U_V \otimes I$, where $\kb0$ acts on a target. Observe that $I - Q_V = U_V(\kb1 \otimes I)U_V \otimes I$, and that $L_2^\prime \ket{\Cat_n,\dov}$ measures to $\bigotimes_{V \in \mathcal I} Q_V \otimes I$ and to $\bigotimes_{V \in \mathcal I} (I - Q_V) \otimes I$ each with probability 1/2. Therefore the result follows from applying \cref{small} with input state $\ket{\iv^\prime}$, circuit $L_1^B$, desired output state $L_2^\prime \ket{\Cat_n,\dov}$, and projections $(Q_V)_V$, recalling that $|L_1^B| \le (c/2)|\mathcal I|$ and that $|\mathcal I| \ge \Omega(n)$.
\end{proof}

\section*{Acknowledgments}
\addcontentsline{toc}{section}{Acknowledgments}

Thanks to Benjamin Rossman and Henry Yuen for introducing me to this problem, and for having several helpful discussions throughout the research and writing processes. Thanks to Srinivasan Arunachalam, Daniel Grier, Ian Mertz, Eric Rosenthal, and Rahul Santhanam for helpful discussions as well. Part of this work was done while the author was visiting the Simons Institute for the Theory of Computing. Circuit diagrams were made using the Quantikz package~\cite{Kay20}.

\printbibliography[heading=bibintoc]

\appendix
\section{Proof of \texorpdfstring{\cref{par-ub} ($d < 11$)}{Corollary 1.2 (d < 11)}}\label{d7}

\begin{SCfigure}[][b]
	\begin{quantikz}
		& \qw & \qw & \control{} & \qw & \qw & \qw \\
		\lstick{\ket0} & \gate \vee & \gate X & \ctrl{-1} & \gate X & \gate \vee & \qw \\
		& \ctrl{-1}& \qw & \qw & \qw & \ctrl{-1} & \qw \\
		& \ctrl{-1}& \qw & \qw & \qw & \ctrl{-1} & \qw \\
		& \ctrl{-1} & \qw & \qw & \qw & \ctrl{-1} & \qw
	\end{quantikz} =
	\begin{quantikz}
		& \control{} & \qw \\
		\lstick{\ket0} & \ghost X \qw & \ghost \vee \qw \\
		& \octrl{-2} & \qw \\
		& \octrl{-1} & \qw \\
		& \octrl{-1} & \qw
	\end{quantikz}
	\caption{}
	\label{simp}
\end{SCfigure}

Here we give only the aspects of the proof that differ from the depth-11 case. The rest of the argument, and a reminder of \cref{par-ub} itself, may be found in \cref{par-ub-sec}.
	
Let $C$ be the depth-2 QAC circuit from \cref{mc-ub-sec}, and let $\ket\nu$ be an $n$-nekomata such that $C\ket\zs$ and $\ket\nu$ have high fidelity. Let $C^\prime = (X^{\otimes n} \otimes I) C$ and $\ket{\nu^\prime} = (X^{\otimes n} \otimes I) \ket\nu$, where $X^{\otimes n}$ acts on the targets of $\ket\nu$. Note that $\ket{\nu^\prime}$ is an $n$-nekomata with the same targets as $\ket\nu$, and that the fidelity of $C^\prime \ket\zs$ and $\ket{\nu^\prime}$ equals that of $C\ket\zs$ and $\ket\nu$.
	
To approximate parity in depth 7, plug $C^\prime$ into the circuit for parity from \cref{pc}, and apply the simplification pictured in \cref{simp}. The gate on the right denotes $R_{\ket{1000}}$ (more generally, consider $R_{\ket{1,\zs}}$), and the equivalence of the circuits in \cref{simp} may be verified by considering their actions when the unset qubits range over all standard basis states.

\section{Proof of \texorpdfstring{\cref{mc-nek-lb-mean}}{Theorem 4.2(i)}}\label{mc-app}

\begin{rst*}[\cref{mc-nek-lb-mean}, paraphrased]
	Let $CL$ be a mostly classical circuit of size $s$ and depth $d$, where $C$ is purely classical and $L$ is a layer of \rtt gates. Then the fidelity of $CL\ket\zs$ and any $n$-nekomata is at most
	\begin{equation*}
		\frac{1}2 + 
		\exp\left(-\Omega\left(\frac{n/(4^d \log n)}{\max\left(\log s, \sqrt{n/(4^d \log n)}\right)}\right)\right).
	\end{equation*}
\end{rst*}

Designate $n$ qubits of $CL$ as targets. Our proof of \cref{mc-nek-lb-mean} is similar to that of \cref{mc-nek-lb-nice}, except that here our procedure for simulating a standard-basis measurement of the targets of $CL\ket{\zs}$ is much more complicated. For brevity's sake we will omit some proof steps with clear analogues in the proof of \cref{mc-nek-lb-nice}, i.e.\  in \cref{mc-lb-sec}.

Consider a gate $G$ in $L$. Write $G = \rt{\bigotimes_j \ket{\cv_j}}$ for one-qubit states $(\ket{\cv_j})_j$, and let $p_j = p_j^{(G)} = |\ip1{\cv_j}|^2$. We may assume that $p_j \neq 0$ for all $j$, because $G\ket{\zs} = (\rt{\bigotimes_{j: p_j \neq 0} \ket{\cv_j}} \otimes I) \ket{\zs}$. Then, as an immediate corollary of \cref{d1-distr}, a standard-basis measurement of $G\ket{\zs}$ is distributed identically to $(B \wedge X_j)_j$, where the $X_j$ are independent $\bern(p_j)$ random variables conditioned on $(X_j)_j$ not being the all-zeros string, and $B \sim \bern\left(4\prod_j(1-p_j) - 4\prod_j(1-p_j)^2\right)$ is independent of $(X_j)_j$.

Let $R = (R_j)_j$ where each $R_j$ is independently 1 with probability $1-p_j$ and uniform random on $[0,1)$ with probability $p_j$. Then $(X_j)_j$ is distributed identically to $(\Ind{R_j < 1})_j$ conditioned on $\min R < 1$. Let $\argmin R$ be a value of $j$ such that $R_j = \min R$, and note that if we condition on $\min R < 1$ then $\argmin R$ is unique with probability 1. To sample $R$ conditioned on $\min R < 1$, one may first sample $J = \argmin R$ conditioned on $\min R < 1$, next sample $\mu = \min R$ conditioned on $R_J = \min R < 1$, and finally, for all $j \neq J$, independently sample $R_j$ conditioned on $R_j > \mu$.

Rather than sampling $(\argmin R \mid \min R < 1)$ (i.e.\ $\argmin R$ conditioned on $\min R < 1$) directly, we may do so as follows. Identifying $C$ with the function from boolean strings to boolean strings that it computes, say that the $j$'th input bit ``influences" the $k$'th output bit if there exist strings $x,y$ differing only in position $j$ such that $C\ket x$ and $C\ket y$ differ in position $k$. Recall that no input bit influences more than $2^d$ output bits. Let $\tau^{(G)}$ be the (non-random) tree constructed in the following two steps:
\begin{itemize}[leftmargin=*]
	\renewcommand\labelitemi{--}
	\item Start with a rooted binary tree with $\binom n {2^d}$ leaves and depth $\ceil{\log \binom n {2^d}}$, and identify each leaf with a distinct set of $2^d$ targets of $CL$.
	\item Then, for each qubit $v$ acted on by $G$, for some set $u$ of $2^d$ targets including all of the targets influenced by $v$, add the node $v$ and edge $(u,v)$ to the tree.
\end{itemize}
For each non-leaf node $u$ in $\tau^{(G)}$ such that $(\argmin R \mid \min R < 1)$ is descended from $u$ with nonzero probability, independently ``highlight" a random edge from $u$ to one of its children, where the probability of highlighting an edge $(u,v)$ equals the probability that $(\argmin R \mid \min R < 1)$ is descended from $v$ divided by the probability that $(\argmin R \mid \min R < 1)$ is descended from $u$. Then there is a unique root-to-leaf path consisting only of highlighted edges, and the leaf at the end of this path is distributed identically to $(\argmin R \mid \min R < 1)$.

Altogether this implies the following procedure for simulating a standard-basis measurement of $G\ket{\zs}$. If $G$ acts on a single qubit then simply output $Y^{(G)} \sim \bern(|\bra1G\ket0|^2)$. Otherwise, first sample the following random variables, all independently:
\begin{itemize}[leftmargin=*]
	\renewcommand\labelitemi{--}
	\item Sample $B^{(G)} \sim \bern\left(4\prod_j(1-p_j^{(G)}) - 4\prod_j(1-p_j^{(G)})^2\right)$.
	\item Highlight random edges in $\tau^{(G)}$, in the manner described above.
	\item For all $j$, sample $M_j^{(G)}$ from the distribution of $\min R$ conditioned on $R_j = \min R < 1$;
	\item For all $j$, sample $S_j^{(G)}$ from the uniform distribution on $[0,1]$.
\end{itemize}
Then let $J^{(G)}$ be the leaf in the root-to-leaf path consisting of highlighted edges in $\tau^{(G)}$, and output
\begin{equation*}
\left(\left(B^{(G)} = 1\right) \wedge \left(\left(J^{(G)} = j\right) \vee \left(S_j^{(G)} \le P\left(R_j < 1 \mid R_j > M^{(G)}_{J^{(G)}}\right)\right)\right)\right)_j.
\end{equation*}

For $0 \le k \le \ceil{\log \binom n {2^d}}$ let $E_k^{(G)}$ be the set of highlighted edges between nodes at depths $k$ and $k+1$ in $\tau^{(G)}$, where we define the root to have depth 0, children of the root to have depth 1, and so on. Note that $(E_k^{(G)})_k$ is a partition of the set of highlighted edges in $\tau^{(G)}$. Let $Y = (Y^{(G)})_G, B = (B^{(G)})_G, E_k = (E_k^{(G)})_G, M = (M_j^{(G)})_{j,G}, S = (S_j^{(G)})_{j,G}$.

Recall that $s$ is defined as the size of $CL$, and assume (without loss of generality, given the theorem we are proving) that $s \ge \exp\left(\sqrt{n/(4^d \log n)}\right)$. Since $\Ex[B^{(G)}] \le 4\exp\left(-\sum_j p^{(G)}_j\right)$ for all $G$, we may assume that $\max_G \sum_j p^{(G)}_j \le 2 \ln s$, by the same reasoning as in \cref{mc-lb-sec}. Call $S$ ``good" if $|\{j: S^{(G)}_j \le p^{(G)}_j\}| \le c \log s$ for all $G$, where $c$ is an appropriately large universal constant. As in \cref{mc-lb-sec}, by a Chernoff bound, the probability that $S$ fails to be good is at most $s^{-\Omega(1)}$. For all $G$,
\begin{equation*}
P\left(R_j < 1 \mid R_j > M^{(G)}_{J^{(G)}}\right) \le
P(R_j < 1) = p_j^{(G)}
\end{equation*}
(where the definition of $R_j$ here implicitly depends on $G$), so if $S$ is fixed and good then there are at most $O(\log s)$ indices $j$ such that the boolean value $S_j^{(G)} \le P\left(R_j < 1 \mid R_j > M^{(G)}_{J^{(G)}}\right)$ is not identically false.

\begin{rst*}[{\cref{glss}}]
	Let $(Y_1, \dotsc, Y_n)$ be a read-$r$ family, and let $\mu = \Ex\left[\sum_{j=1}^n Y_j\right]$. Then for all $\eps\ge0$,
	\begin{align*}
		P(Y_1 + \dotsb + Y_n \ge \mu + \eps n) &\le \exp(-2\eps^2 n/r), \\
		P(Y_1 + \dotsb + Y_n \le \mu - \eps n) &\le \exp(-2\eps^2 n/r).
	\end{align*}
\end{rst*}

Let $V = (Y, B, (E_k)_k, M, S)$, and note that the targets of a standard-basis measurement of $CL\ket{\zs}$ are a read-$O(2^d \log s)$ family if $V\backslash Y$ is fixed, or if $V\backslash B$ is fixed and $S$ is good, or if $V\backslash M$ is fixed and $S$ is good, or if $V\backslash S$ is fixed. (We will consider the case where $V\backslash E_k$ is fixed and $S$ is good shortly.) Let $W = W(V)$ be the Hamming weight of a standard-basis measurement of the targets of $CL\ket\zs$. Then, by \cref{nek-lb}, \cref{glss}, and an argument involving the triangle inequality\footnote{In slightly greater detail: sample an independent copy $V^\prime$ of $V$, use a hybrid argument to show that $|W(V) - W(V^\prime)|$ is small with high probability, and then fix a value of $V^\prime$ such that $W(V)$ is concentrated around $W(V^\prime)$.} similar to that in \cref{mc-lb-sec}, it suffices to prove the following:
\begin{clm}\label{clm}
	Fix $V \backslash (E_k)_k$ such that $S$ is good, and let $\mu = \Ex[W \mid V \backslash (E_k)_k]$. Then for all $\eps > 0$,
	\begin{align*}
	P(W \ge \mu + \eps n) \le \exp(-\Omega(\eps^2 n/(4^d \log s \log n))), \\
	P(W \le \mu - \eps n) \le \exp(-\Omega(\eps^2 n/(4^d \log s \log n))),
	\end{align*}
	where the probabilities are over $(E_k)_k$.
\end{clm}

Observe that for all $k$, if $V\backslash E_k$ is fixed and $S$ is good then the targets of a standard-basis measurement of $CL\ket{\zs}$ are a read-$O(2^d\log s)$ family. Before proving \cref{clm}, we remark that a similar statement\footnote{$P(|W - W^\prime| \ge \eps n) \le O(2^d \log n) \cdot \exp(-\Omega(\eps^2n/(8^d\log s\log^2 n)))$, where $W^\prime$ is an independent copy of $W$. This may be proved by writing $W-W^\prime$ as a sum of $\Theta(\log \binom n {2^d})$ terms that are each required to be of magnitude $O(\eps n/\log \binom n {2^d})$.} with weaker parameters can be proved using another similar argument involving the triangle inequality.

Let $\gauss(v)$ be the set of real-valued random variables $X$ such that $\Ex [e^{\lambda X}] \le \exp(\lambda^2 v/2)$ for all $\lambda \in \R$. (This definition is motivated by the fact that if $X$ is Gaussian with mean 0 and variance $v$ then $\Ex [e^{\lambda X}] = \exp(\lambda^2 v/2)$ for all $\lambda \in \R$.) Boucheron, Lugosi and Massart~\cite{BLM13} noted that random variables obeying ``sub-Gaussian" tail bounds also have sub-Gaussian moment generating functions, and vice versa:

\begin{lem}[{\cite[Chapter 2.3]{BLM13}}]
	Let $X$ be a real-valued random variable such that $\Ex[X] = 0$.
	\begin{enumerate}[label=(\roman*), ref={\thelem(\roman*)}]
		\item If $\max(P(X>t), P(X<-t)) \le \exp(-t^2/(2v))$ for all $t>0$, then $X \in \gauss(16v)$. \label[lem]{gauss-i}
		\item If $X \in \gauss(v)$, then $\max(P(X>t), P(X<-t)) \le \exp(-t^2/(2v))$ for all $t>0$. \label[lem]{gauss-ii}
	\end{enumerate}
\end{lem}

The following lemma is basically implicit in the martingale proof of McDiarmid's inequality~\cite{BLM13}, and is proved below for completeness:

\begin{lem}\label{gauss-mgf}
	Let $X_1, \dotsc, X_m$ be independent random variables, let $v_1, \dotsc, v_m > 0$, and let $f$ be a function such that
	\begin{equation*}
	f(x_1, \dotsc, x_{i-1}, X_i, x_{i+1}, \dotsc, x_m) - \Ex f(x_1, \dotsc, x_{i-1}, X_i, x_{i+1}, \dotsc, x_m)  \in \gauss(v_i)
	\end{equation*}
	for all $i\in[m]$ and fixed $x_1, \dotsc, x_{i-1}, x_{i+1}, \dotsc, x_m$. Then,
	\begin{equation*}
	f(X_1, \dotsc, X_m) - \Ex f(X_1, \dotsc, X_m)  \in \gauss\left(\sum_{i=1}^m v_i\right).
	\end{equation*}
\end{lem}

\begin{rmk*}
	\cref{gauss-mgf} is tight when $(X_i)_i$ are independent Gaussians and $f$ is the summation function.
\end{rmk*}

\begin{proof}[Proof of \cref{clm} assuming \cref{gauss-mgf}]
	It follows from \cref{glss,gauss-i} that $W - \Ex[W \mid V \backslash E_k] \in \gauss(O(n2^d\log s))$ for all $k$. Since $\log \binom n {2^d} \le 2^d \log n$, it then follows from \cref{gauss-mgf} that $W - \Ex[W \mid V \backslash (E_k)_k] \in \gauss(O(n4^d \log n \log s))$. Finally, apply \cref{gauss-ii}.
\end{proof}

\begin{proof}[Proof of \cref{gauss-mgf}]
	Let $Y = f(X_1, \dotsc, X_m)$, and for $i \in [m]$ let $X_{[i]} = (X_j)_{j \le i}$ and $E_i = \Ex[Y \mid X_{[i]}]$. Fix $\lambda \in \R$, and let $\varphi(x) = e^{\lambda x}$. We will prove that $\Ex \varphi(E_i - E_0) \le \exp(\lambda^2 v_i/2) \Ex \varphi(E_{i-1} - E_0)$ for all $i \in [m]$, from which it follows by induction that
	\begin{equation*}
	\Ex \varphi(Y - \Ex Y) = \Ex \varphi (E_m - E_0) \le \exp\left(\lambda^2 \sum_i v_i/2\right),
	\end{equation*}
	as desired. Since
	\begin{equation*}
	\Ex \varphi(E_i - E_0) =
	\Ex \Ex[\varphi(E_i - E_0) \mid X_{[i-1]}] =
	\Ex[\varphi(E_{i-1} - E_0) \Ex[\varphi(E_i - E_{i-1}) \mid X_{[i-1]}]],
	\end{equation*}
	it suffices to prove that $\Ex[\varphi(E_i - E_{i-1}) \mid X_{[i-1]}] \le \exp(\lambda^2 v_i/2)$. Let $X_{\backslash i} = (X_j)_{j \neq i}$ and $E_{\backslash i} = \Ex[Y \mid X_{\backslash i}]$. By Jensen's inequality,
	\begin{equation*}
	\varphi(E_i - E_{i-1}) =
	\varphi(\Ex[Y - E_{\backslash i} \mid X_{[i]}]) \le
	\Ex[\varphi(Y - E_{\backslash i}) \mid X_{[i]}],
	\end{equation*}
	so
	\begin{equation*}
	\Ex[\varphi(E_i - E_{i-1}) \mid X_{[i-1]}] \le
	\Ex[\varphi(Y - E_{\backslash i}) \mid X_{[i-1]}] \le
	\sup_{X_{\backslash i}} \Ex[\varphi(Y - E_{\backslash i}) \mid X_{\backslash i}] \le
	\exp(\lambda^2 v_i/2). \qedhere
	\end{equation*}
\end{proof}

\end{document}